\documentclass{amsproc}
\usepackage{amsmath,amssymb,amsthm}
\usepackage[shortcuts]{extdash}
\usepackage{hyperref}

\numberwithin{equation}{section}
\theoremstyle{definition}

\usepackage{csquotes}
\usepackage[style=numeric,backend=biber,maxbibnames=99,sorting=nyt]{biblatex}
 \bibliography{wpf.bib}

\usepackage{color}
\usepackage{lipsum}
\usepackage{amsfonts,epic,eepic,float}

\usepackage[mathscr]{eucal}

\usepackage{rotating,epsfig,indentfirst,array,varioref}
\usepackage{appendix,marginnote,tikz,pgf,mathtools}
\usepackage{setspace}

\usepackage{tikz,tikz-cd,pgf}
\usetikzlibrary{shapes}
\usetikzlibrary{calc}
\usetikzlibrary{decorations.pathmorphing}
\usetikzlibrary{decorations.pathreplacing,shapes.misc}
\usetikzlibrary{positioning}
\usetikzlibrary{arrows}
\usetikzlibrary{decorations.markings}

\def\be{\begin{equation}}
\def\ee{\end{equation}}
\def\barr{\begin{array}}
\def\earr{\end{array}}

\def\al{\alpha}

\def\1{\tilde{1}}
\def\2{\tilde{2}}
\def\3{\tilde{3}}

\newtheorem{theorem}{Theorem}[section]
\newtheorem{Proposition}{Proposition}[section]
\newtheorem{Corollary}{Corollary}[section]
\newtheorem{remark}[theorem]{Remark}
\newtheorem{lemma}[theorem]{Lemma}
\newtheorem{definition}[theorem]{Definition}

\numberwithin{equation}{section}

\newcommand{\ba}{\begin{equation}\begin{aligned}}
\newcommand{\ea}{\end{aligned}\end{equation}}

\newcommand{\bml}{\begin{multline}}
\newcommand{\eml}{\end{multline}}

\newcommand{\nablab}{\nabla \!\!/}

\newcommand{\CC}{\mathbb{C}}
\newcommand{\RR}{\mathbb{R}}

\newcommand{\ZZ}{\mathbb{Z}}

\newcommand{\dd}{\mathrm{d}}
\newcommand{\pd}{\partial}

\newcommand{\TT}{\mathcal{T}}
\newcommand{\BB}{\mathcal{B}}

\newcommand{\OO}{\mathcal{O}}
\newcommand{\MM}{\mathscr{M}}

\newcommand{\LL}{\mathscr{L}}

\newcommand{\HH}{\mathscr{H}}

\begin{document}

\title {Primitive Forms without Higher Residue Structure\\ and\\ \ \ \ Integrable Hierarchies\ (I)}

\author{Konstantin Aleshkin}
\address{Columbia University \\ 2990 Broadway New York, NY 10027}
\email{aleshkin@math.columbia.edu}

\author{Kyoji Saito}
\address{
Research Institute for Mathematical Sciences, Kyoto University\\
 Sakyoku Kitashirakawa, Kyoto, 606-8502, Japan\\
  Institute for Physics and Mathematics of the Universe \\
 the University of Tokyo, 5-1-5 Kashiwanoha, Kashiwa, 277-8583, Japan \\ 
 Laboratory of AGHA, Moscow Institute of Physics and Technology\\
 9 Institutskiy per., Dolgoprudny, Moscow Region, 141700, Russian Federation.}
\email{kyoji.saito@ipmu.jp}

\subjclass[2000]{Primary}

\date{}

\begin{abstract} We introduce {\it primitive forms with or  without higher residue
structure} and explore their connection with
the flat structures with or without a metric~\cite{Sabbah} and integrable hierarchies of KdV type. Just as the classical case of primitive forms with metric~\cite{LLS}, the primitive forms without metrics are constructed as the positive part of the Birkhoff decomposition of formal oscillatory integrals with respect to the descendent variable. The oscilating integrals of a primitive form without metric give rise to a hierarchy of commuting  PDE of the KdV type as inthe case of primitive forms with metric. This shall be studied in (II).

\end{abstract}

\maketitle


  \section{Introduction}
  We introduce {\it primitive forms without higher residue
  structure} (or the same to say, {\it primitive forms without metric structure})\footnote
  {In the present paper,  we shall use a convention that the expression ``without metric" includes also the cases  ``with metric".}
  and explore their connection with
  the flat structures without a metric and integrable hierarchies of the KdV type.

  We start with a historical overview on primitive forms.
  Primitive forms were introduced by the second named author as the integrands for vanishing cycles in a universal unfolding family of a hypersurface isolated singularity. They are formulated in terms of higher residue pairings~\cite{HigherResidue83}. Namely, a primitive form is a top degree relative de Rham
  cohomology class on a certain family of (open) hypersurfaces, satisfying a system of
  bilinear relations described by the higher residue pairings~\cite{HigherResidue83, LLS}. In particular, the conditions include a purity condition of the higher residue pairing on the first derivatives of the primitive form. Then, those
  bilinear relations on a primitive form induce a geometric structure on the parameter space of the family, called the flat
  structure~\cite{FromPrimitiveToFrobenius}. That is, the tangent bundle of the parameter space obtains an algebra structure compatible with a
  flat metric on it. This implies the existence of a flat coordinate system on the parameter space and a potential function with respect to the triple tensor product of vector fields. Actually,
  such structure on the parameter space has been found already before on the orbit space
  of a finite reflection group (\cite{SReflection, SYS}) and was called the flat structure. Later a similar structure was found in 2D-topological field theories~\cite{Dubrovin}, and was 
  axiomatized by Dubrovin as a Frobenius manifold structure~\cite{Dubrovin, Manin, Hertling}.

  Sabbah~\cite{Sabbah} and Manin~\cite{Manin} considered a generalized version of that flat 
  
flat structure without metric got a new interest and impetus since the structure appeared on the deformation parameter space of Okubo-type ordinary differential equations, and consequently, one finds a Saito structure without metric on the orbit space of a unitary reflection
group~\cite{KMS1, KMS2, KMS3} (c.f. also~\cite{KoMi}). Similar structures appear in the study of relative and open Gromov-Witten invariants~\cite{Solomon:2019lzb, BuryakOpenSaitoTheory}

On the other hand, the primitive forms, as mentioned already, to be considered as integrands for vanishing cycles, can be re-considered also as the integrands over Lefschetz thimbles for oscillatory integrals for the unfolding family of defining function of the hypersurface singularity.\footnote{
  To be exact, there is a one dimensional difference between vanishing cycles and Lefschetz thimles, and accordingly a degree one difference between the corresponding integrands. This can be explained by the relation of a top-form on an ambient space and its Poincare residues  on hypersurface-fibers. Here, in the universal unfolding, the initial function $f$ for the unfolding and its defining domain $X$, is called the {\it generating center} according to R.\! Thom \cite{Thom}.}
Then, primitive forms can be captured as the positive part in the {\it descendent variable}\footnote
{We mean by the {\it descendent variable} the denominator variable $z$ in the exponential of the oscillatory integral factor.
} of the Birkhoff decomposition of the evolution image by an oscillatory integral factor action 
on the top degree form on the generating center belonging to a  suitable section of filtered relative De Rham cohomology group of the generating center~(\cite{LLS} ). This view point naturally reduces the defining conditions on a primitive form to the conditions on the pair of the top form and the section on the generating center to be {\it good}. In particular, the purity condition of higher residue pairings on the primitive forms is reduced to a purity condition of the classical residue pairing defined on the section.  

\bigskip
The present paper  is the first part in a series of two papers. In the present  Part I, we introduce a concept of a primitive form without higher residue structure (that is, we remove a purity condition on the primitive forms) \footnote{  Actually this version of primitive forms without higher residue 
structure was already discussed by  some authors, but was not written explicitly (Lectures by the second author in
Ecole Polytechnique, Neimehen  and Nice (1980). See also~\cite{BaKo}).  
} and discuss their connection to the flat structure without metric.

We first recall in \S2-4 some
basic frame works and techniques on primitive forms.  Namely, in \S2, for a
given unfolding $F(\underline{x},\underline{t})$ of a function
$f(\underline{x})$ with an isolated critical point (=the generating center), we attach the filtered De
Rham cohomology groups (that is, the cohomology theory, which contains, so
called, the {\it descendent variable $z$}). Then in \S3 the concepts of sections and
opposite filtrations on the filtered de-Rham cohomology groups are recalled. In \S4, we recall
the technique of formal analysis (i.e.\ perturbation theory) of the deformation
of oscillatory integrals, developed in \cite{LLS}. That is, we formally study
the actions of the oscillatory integral factor
$e^{(f(x)-F(\underline{x}.\underline{t}))/z}$ which intertwines  the twisted de Rham cohomology group for $f$  with that for $F$. The
main difficulty here is the {\it mixture of the descent variable $z$ and the
deformation variable $\underline{t}$}, since, in the oscillatory integral
factor, the descendent variable $z$ has essential singularity at $z=0$. However,
the perturbative approach allows us to treat only finite order of poles in the
descent variable $z$ so far as we remain in a formal neighborhood of the origin of the deformation space $S$ of bounded order. 

After these preparations, the definitions of primitive
forms without metric structures are given in \S5.  As in the classical case \cite{LLS}, their formal construction is achieved as follows. We first consider a good section without metric (which is equivalent data of a good opposite filtration without metric)\footnote
{We mean by ``without metric" that we remove the purity condition on the residue pairing on the section in the definition.} to the filtered De Rham cohomology group of the generating center. Then a pair of such section without metric and a top degree De Rham cohomology class in the section is called a good pair without metric if they satisfy certain primitivity and homogeneity conditions. Then the primitive form without metric is given by the 
positive power part w.r.t.~the descendent variable, i.e.\! the positive part of the Birkhoff decomposition, of the evolution image by the oscillatory integral factor action on the top form. We show in \S6 that there is a one to one correspondence between (a) the set of primitive forms without  metric structure  and (b) the set of  good pairs without metric condition.  


The construction of a flat structure without metric from a primitive form without metric is  an easy formal procedure which is parallel to the classical case \cite{S2} just taking care of the lack of the purity condition on the residue pairing, given  in \S7.  
We compare the flat structure without metric with the already known definitions
of Saito structure without metric and F-manifolds.


\medskip
In the coming Part II, we explore the connection of these primitive forms without metric structure with integrable hierarchies
of the KdV type.  
Let us outline how the integrable hierarchy is constructed from a primitive form and the
corresponding flat structure. 
We note that the evolution by the oscillatory integral factor preserves opposite filtration both in the filtered De Rham cohomology groups for the generating center and for the unfolding. Therefore, the inverse evolution of the primitive form gives the prinicipal part of the Laurent expansion with respect to the descendent variable in the generating center. That is, the oscillatory integral of the primitive form given by the inverse evolution map applied to the primitive form is a formal power series in the inverse powers of the descendent variable $z$. As a consequence of the properties
of the primitive form, those coefficients of  the inverse oscillatory integral satisfies a certain system of differential equations. These equations are crucial in the construction of integrable hierarchies.
%
%
The decomposition coefficients are called deformed flat coordinates, since the first coefficients in
the descendent variable coincide with actual flat coordinates.
The particular choice of such a set of deformed flat coordinates defines a calibration of
the flat structure (without metric).

In the case where the primitive form has a compatible metric structure, the corresponding flat structure is a Frobenius manifold.
There is a well-known construction that associates an integrable system to a Frobenius manifold
with an additional data called calibration~\cite{Dubrovin}.
Such an integrable system is called the principle
hierarchy associated to a given calibrated Frobenius manifold. The equations of the principle
hierarchy are defined on an appropriate version of the loop space of the Frobenius manifold.
The most notable example is the (dispersionless) KdV hierarchy which is constructed on the
one-dimensional deformation space of the $A_1$-singularity. The KdV hierarchy is a system of PDE
on one function. The range of this function is identified with the deformation space of the
singularity.


In the case of primitive forms without metric the flat structure on the deformation
space is basically a Saito structure without metric or a conformal flat F-manifold.
It is known in the literature~\cite{Lorenzoni1, Lorenzoni2}, that such structures give rise to integrable
systems of PDE as well via a very similar construction. However, the resulting PDE are
not necessarily Hamiltonian which is due to the fact that the flat structure has no natural metric.
We remark that there is another construction of integrable systems from generalized
Frobenius-type structures which arise from quantum K-theory~\cite{TodorK}. In the quantum K-theory
case a Frobenius structure has a flat metric but lacks flat identity and Euler vector fields (see
the main text).

The second paper is devoted the explanation of the structure of
the non-Hamiltonian infinite commuting PDE that can be constructed from primitive forms
without metric structure.
The defining equations of the integable hierarchy
are also constructed using the oscillatory integrals of the
primitive form. We also discuss generalizations of the
properties of integrable hierarchies from the case of primitive forms with metric structure.




\section{Filtered De Rham cohomology modules $\HH_F$}

In this section, we introduce the semi-infinite filtered De Rham cohomology
modules associated with an isolated critical point of a holomorphic function
(cf.~\cite{Hertling,S1,HigherResidue83,FromPrimitiveToFrobenius,LLS}).

\medskip Let $X$ be an open Stein and contractible neighborhood of the origin
$0$ in the complex space $\CC^{n+1}$ ($n\in \ZZ_{\ge0}$). Let $f$ be a holomorphic
function on $X$ whose critical point is only at the origin.  

\begin{definition}
\label{Unfolding} A tuple $(Z,S,p)$ is called a \textit{frame} of an unfolding $F$ of
$f$ (or, $F$ is an unfolding of $f$ over a frame $(Z,S,p)$), if it satisfies the following conditions 1., 2. and 3.

1.  The space $Z$ is a Stein open neighborhood of $0\in \CC^{n+m+1}$
($m\in\ZZ_{\ge0}$), $S$ is a Stein open neighborhood of the base point $0\in
\CC^m$, $p:Z\to S$ is the natural projection with $p^{-1}(0)=X$.

2.  The $F$ is a holomorphic function on $Z$ such that the restriction $F\mid
_{X=p^{-1}(0)}$ coincides with $f$.

3. The restriction $p\mid_{C_F}$ of the projection $p$ to the relative critical
set $C_F$ (see \eqref{eq:Critical} below) is finite.

\smallskip In particular, $F=f$ itself is a trivial unfolding defined over the
frame $(X,S=\{0\},p_0)$ where $p_0$ is the trivial projection $X\to \{0\}$.
\end{definition}

We allow to change and shrink $X, \ S$ and $Z$ suitably if necessary. For a
convenience, we assume further an auxilary data  that there is a projection 
$$
\pi_X\ : \ Z\ \longrightarrow\  X
$$ 
so that $Z$ is embedded into $X\times S\subset \CC^{n+1}\times \CC^{m}$ as a domain by the map $(\pi_X,p)$. 
We shall denote by
$(x_0,\ldots,x_n)$ the coordinate system of $\CC^{n+1}$ and by $t_1,\ldots,t_m$ the
coordinate system of $\CC^\mu$. Partial derivatives with respect to the
coordinate system are denoted by $\partial_{x_0},\ldots,\partial_{x_n}$ and
$\partial_{t_1},\ldots,\partial_{t_m}$, respectively.  Under this notation, the
relative critical set is defined by 
\be
\label{eq:Critical} 
C_F:=\{\partial_{x_0}F=\ldots=\partial_{x_n}F=0\} \quad
\text{and}\quad \OO_{C_F}:=\OO_Z/(\partial_{x_0}F, \ldots,\ \partial_{x_n}F).
\ee 
The relative critical set $C_F$ is Cohen-Macaulay variety such that the
projection $p\mid_{C_F}$ is finite $\mu$-ramified covering over $S$. Actually,
the direct image $p_*(\OO_{C_F})$ is locally free $\OO_S$-module of rank
$\mu:=\dim_\CC \OO_X/(\partial f)$, called the Milnor number of $f$. We define a
natural $\OO_S$-homomorphism, 
\be
\label{eq:KS} 
\mathcal{T}_S\ \to\ p_*\mathcal{O}_{C_F}, \qquad v \ \mapsto \
\tilde{v}F|_{C_F} 
\ee
where $\mathcal{T}_S$ is the sheaf of holomorphic tangent vector fields on $S$ and $\tilde v$ is any lifting of a vector
field $v$ on $U\subset S$ to that on $p^{-1}(U)\subset Z$ such that $p_*(\tilde{v})=v$. Then, the restriction
$\tilde{v}F|_{C_F\cap p^{-1}(U)}$ does not depend on the choice $\tilde v$.

\medskip
\noindent
\begin{definition}
\label{UniversalUnfolding} An unfolding $F$ is called \textit{universal} if the morphism
\eqref{eq:KS} is an $\OO_S$-isomorphism (\cite{Thom}).  
\be
\label{eq:KS2} 
\mathcal{T}_S\ \simeq \ p_*\mathcal{O}_{C_F} 
\ee 
\end{definition}

The definition implies, in
particular, the dimension $\dim S=:m$ is equal to the Milnor number $\mu$ of
$f$.  
Using this isomorphism, we particularly introduce global vector fields on $S$ 
\be 
\label{eq:UnitEuler0}
  \qquad   \partial_1 \quad \text{and} \quad E \quad  \in \quad \Gamma(S,\TT_S)
\ee    
called the {\it unit vector field} and the {\it Euler vector field}, respectively, as those
corresponding to the unit
1 and the class of the function $F$ in the RHS of  \eqref{eq:KS2}, respectively.  That is, we have 
\be
\label{eq:UnitEuler} 
\tilde{\partial}_1 F\mid_{C_F} = 1 \quad {\text and}\quad
\tilde{E}F\mid_{C_F} = F \mid_{C_F}.  
\ee
(See \cite{QH} for a justification of this notion).

We usually choose $ \tilde{\partial}_1$ in such manner that $ \tilde{\partial}_1
F=1$.  We shall often use coordinates system $t_1,\ldots,t_m$ such that
$\partial_1=\partial_{t_1}$.  The tangent bundle $\TT_S$ of $S$ gets an
$\OO_S$-algebra structure through the isomorphism \eqref{eq:KS2}, where we
denote by $*$ the product, and $\partial_1$ plays the role of the unit element.  
\be
\label{eq:*}
\begin{array}{cl} * 
\quad:\quad \TT_S\ \times \TT_S \ \longrightarrow \ \TT_S &
\\ \\ (\widetilde{v_1* v_2} ) (F)\mid_{C_F} := \tilde{v_1}(F)\mid_{C_F} \cdot \
\tilde{v_2}(F)\mid_{C_F} \quad& \text{for } v_1,v_2 \in \TT_S.
\end{array} 
\ee 

Next, we introduce the filtered de-Rham cohomology groups
\eqref{eq:FilteredDeRham1} and \eqref{eq:FilteredDeRham2} associated with any
unfolding $F$ of $f$. Let $\Omega^*_{Z/S} = \oplus_{i=0}^{n+1}\Omega_{Z/S}^i$,
where $\Omega_{Z/S}^i:=\Omega^i_Z/p^*\Omega_S^1\wedge \Omega^{i-1}_Z$ ($i=0,\ldots,n+1$), be the
sheaf of the exterior algebra of holomorphic relative de Rham differentials on
$Z$ relative to the map $p$. We consider two structures of complex on it: the relative de
Rham complex $(\Omega^*_{Z/S}, \dd=\dd_{Z/S})$ over the base space $S$ and the
Poincare complex $(\Omega^*_{Z/S}, dF\wedge)$, which are anti-commuting to each
other.
Combining the two complexes linearly, we introduce the key complex (\cite{HigherResidue83})
of the present paper: 
\be
\label{eq:D_F} (\Omega_{Z/S}^*((z)),D_F).  
\ee 
Here (1) $\Omega^*_{Z/S}((z)):=
\underset{\infty\leftarrow k}{\lim} (\Omega^*_{Z/S}\otimes\CC[z,z^{-1}])/z^k
(\Omega^*_{Z/S}\otimes\CC[z])$ ($z$ is a formal variable commuting with the
differentials $\dd$ and $dF\wedge$. See Footnote 1), and (2) the differential $D_F$ of the
complex is the linear combination 
\begin{equation}
  D_F\ := \ z\dd_{Z/S} + dF\wedge \quad \text{such
    that} \quad (D_F)^2=0 
\end{equation}

Since (i) the relative critical points of $F$ are isolated, (ii) the following
sequence \eqref{eq:Poincare} is exact 
\begin{equation}
  \label{eq:Poincare} 
  0\to \OO_Z \overset{dF\wedge}{\longrightarrow}
  \Omega_{Z/S}^1 \overset{dF\wedge}{\longrightarrow}\Omega_{Z/S}^2
  \overset{dF\wedge}{\longrightarrow} \cdots \overset{dF\wedge}{\longrightarrow}
  \Omega_{Z/S}^{n} \overset{dF\wedge}{\longrightarrow}\Omega_{Z/S}^{n+1} \rightarrow
  \Omega_{Z/S}^{n+1} /\Omega_{Z/S}^{n}\to 0 
\end{equation}
(De Rham Lemma
\cite{deRham}\cite{S4}) and (iii) the morphism $p$ is Stein, we see that the direct
image $\RR p_*$ of the complex \eqref{eq:D_F} is pure $n+1$-dimensional
over $\OO_S((z))$ \cite{HigherResidue83, FromPrimitiveToFrobenius}. That is, only $n+1$-th cohomology group is non-trivial.  Thus we set
\begin{equation}
  \label{eq:FilteredDeRham1} 
  \HH_F:= \RR^{n+1} p_*(\Omega_{Z/S}^*((z)),D_F)
  =p_*(\Omega^{n+1}_{Z/S}((z)))/D_F(p_*(\Omega^{n}_{Z/S}((z)))) 
\end{equation}
which is
naturally a $\OO_S((z))$-module. The cohomology class of $s\in
p_*(\Omega^n_{Z/S}((z)))$ is denoted by $[s]$.

We also define the $k$-th filter $(F^k\Omega_{Z/S}^*((z)),D_F)$ of the complex
\eqref{eq:D_F} by $F^k\Omega_{Z/S}^*((z)):=z^{-k}\Omega_{Z/S}^*[[z]]$ for $k\in \ZZ$.
\footnote{
In the early formulations~\cite{S1, HigherResidue83, FromPrimitiveToFrobenius} of the filtered de-Rham module $\HH_F$, instead of the variable $z$ as in the present paper, we used the variable $\delta_1^{-1}$ in order to indicate that it is a pseudo differential operator (or  Fourier integral operator) associated to the inverse of the primitive derivation $\delta_1=\partial_1$. We switch the view point and regard $\HH_F$ as a module over the space of  the Fourier dual variable $z$ of the primitive derivation $\delta_1$ (this  is the same for the variable $t$ in~\cite{LLS}). For this reason, we have an unfortunate opposite sign convention  between the index $k$ for the filtration $F^k\Omega_{Z/S}^*((z))$ and the exponent $-k$ of $z$ in $z^{-k}\Omega_{Z/S}^*[[z]]$.
}
 Their
cohomology groups are also pure $n+1$-dimensional over $\OO_S[[z]]$ ([ibid]), and we set 
\begin{equation}
  \label{eq:FilteredDeRham2} 
  \HH_F^{(k)}:= \RR^{n+1} p_*(F^k\Omega_{Z/S}^*((z)),D_F) =
  z^{-k}p_*(\Omega^{n+1}_{Z/S}[[z]])/D_F(z^{-k}p_*(\Omega^{n}_{Z/S}[[z]]) 
\end{equation}
which is
naturally a $\OO_S[[z]]$-module.
Then, we see the following semi-infinitely filtered structure on $\HH_F$
together with the Gauss-Manin connection and the higher residue pairings $K_F$
(\cite{S1,S3})

\begin{Proposition} \label{GMHR} 1. The inclusion $F^k\Omega^*_{Z/S}((z))\subset
\Omega^*_{Z/S}((z))$ for $k\in\ZZ$ induces an embedding (injection) $\HH^{(k)}_F
\subset \HH_F $ as $\OO_S[[z]]$-submodule so that we obtain an increasing
filtration in $\HH_F$ : 
\begin{equation}
  \label{eq:Filtration} 
  \cdots\subset \HH^{(-1)}_F\subset \HH_F^{(0)} \subset
  \HH_F^{(1)} \subset \cdots\subset \HH_F
\end{equation}
with the regularity property:
\begin{equation}
  \label{eq:Filter1'} \cup_{k\in\ZZ} \ \HH^{(k)}_F = \HH_F \text{\quad and \quad }
  \cap_{k\in \ZZ} \HH^{(k)}_F = \{0\}.
\end{equation}
The multiplication by $z$ induces $\OO_S[[z]]$-isomorphisms between the filters:
\begin{equation}
  \label{eq:Filter1} z \ : \ \HH_F^{(k+1)} \ \simeq \ \HH_F^{(k)}.  
\end{equation}
That is, $z$ is homogeneous of degree -1 with respect to the filtration (see Footnote 1).

The graded
pieces of the filtration are given by the exact sequence:
\begin{equation}
\label{eq:Filter2} 0\to \HH^{(k-1)}_F \longrightarrow \HH_F^{(k)}
\overset{r_F^{(k)}}{-\!\!\!-\!\!\!\longrightarrow} \Omega_F \to 0
\end{equation}
where
$$ 
\Omega_F:=p_*(\Omega_{Z/S}^{n+1} /dF\wedge\Omega_{Z/S}^{n}) \simeq p_*(\Omega_{Z/S}^{n+1})
/ p_*(dF)\wedge p_*(\Omega_{Z/S}^{n})
\
\space
$$ 
is an $\OO_S$-free module of rank $\mu$, and $r^{(k)}$ is the natural quotient morphism modulo $dF\wedge\Omega_{Z/S}^{n}$ on the leading term of 
$z^{-k}p_*(\Omega_{Z/S}^{n+1})[[z]]$ to $\Omega_F$ such that $r_F^{(k)}=r_F^{(0)}\circ
z^{-k}$.

\medskip 2.  There is a covariant differentiations acting on $\HH_F$, called the
{\it Gauss-Manin connection}:
\begin{equation*}
  \nabla \ :\ \mathcal{T}_{\mathbf{A}_z^1\times S} \times \HH_F \ \longrightarrow
  \ \HH_F
\end{equation*}
given, for $\zeta=[\phi(\underline{x},z)dx_1\wedge\cdots\wedge dx_n]\in \HH_F$,
by
\begin{equation}
  \begin{aligned}
    \label{eq:Gauss-Manin} 
    \nabla_{\frac{d}{dz}}\zeta
    =\Big[\Big(\frac{d\phi}{dz}-z^{-2}F\phi\Big)dx_1\wedge\cdots\wedge dx_n\Big],\\
    \nabla_{v}\zeta =\Big[\Big(\tilde{v}\phi +
    z^{-1}\tilde{v}F\phi\Big)dx_1\wedge\cdots\wedge dx_n\Big] 
  \end{aligned}
\end{equation}
for
$\frac{d}{dz}\in \mathcal{T}_{\mathbf{A}_z^1}$ and $v\in \mathcal{T}_{S}$, such
that

i) The connection $\nabla$ is integrable: $\nabla^2=0$. In particular, $\nabla_{\frac{d}{dz}}$ and $\nabla_{v}$ commute each other.

ii) The connection $\nabla$ over $\mathcal{T}_S$ satisfies the transversality:
$\nabla_{\mathcal{T}_S}: \HH_F^{(k-1)} \to \HH_F^{(k)}$.

iii) The covariant action of $\nabla_{\partial_z}$ has pole of order 2 at $z=0$ \\
\qquad\qquad  (that is, $z^2\nabla_{\partial_z}=-\nabla_{\partial_{z^{-1}}}$ preserves the filteration \eqref{eq:Filtration}).

iv) The operator $\nabla_{z\partial_z+E}$ on $\HH_F$ preserves the filtration \eqref{eq:Filtration}.

\medskip 3.  There is a skew-$\OO_S((z))$-bilinear form, called the higher
residue pairing: \be
\label{eq:HRP} K_F\ : \ \HH_F \ \times \ \HH_F \ \longrightarrow \ \OO_S((z))
\ee such that

i) For $\zeta_1,\zeta_2 \in \HH_F$, we have $K_F(\zeta_1,\zeta_2)=K_F(\zeta_2,
\zeta_1)^*$.

ii) For $\zeta_1,\zeta_2 \in \HH_F$ and $P\in \OO_S((z))$, we have
$$
PK_F(\zeta_1,\zeta_2)=K_F(P\zeta_1, \zeta_2)= K(\zeta_1,P^*\zeta_2).
$$
where we denote by the upper script of $*$ an $\OO_S$-involution of $\OO_S((z))$
defined by $z\mapsto -z$.

iii) For $\zeta_1,\zeta_2 \in \HH_F$ and $v\in \TT_S$, we have
\begin{equation}
  \label{eq:HRv}
  \partial_v K_F(\zeta_1,\zeta_2)=K_F(\nabla_v\zeta_1, \zeta_2) + K_F(\zeta_1,
  \nabla_v\zeta_2) .
\end{equation}

iv) For $\zeta_1,\zeta_2 \in \HH_F$, we have 
\begin{equation}
  \label{eq:HRz}
  z\partial_z K_F(\zeta_1,\zeta_2)=K_F(\nabla_{z\partial_z}\zeta_1,
  \zeta_2) + K_F(\zeta_1, \nabla_{z\partial_z}\zeta_2).
\end{equation}

v) The restriction $K_F \mid_{\HH_F^{(0)}\times \HH_F^{(0)}}$ gives an $\OO_S[[z]]$-linear  map $ \HH_F^{(0)}  \times \HH_F^{(0)} \rightarrow z^{n+1}\OO_S[[z]] $
which makes the following commutative diagram:
\begin{equation}
  \label{eq:HRdeg}
  \begin{array}{rcll}
    \HH_F^{(0)} \ \times \ \HH_F^{(0)}\!\!\!\! & \overset{K_F}{ \longrightarrow } & z^{n+1}\OO_S[[z]] \\
    \\
    r_F^{(0)}\times r_F^{(0)}\  \downarrow   \qquad\qquad &&\quad \quad \downarrow \quad \bmod \ z^{n}\OO_S[[z]] \\
    \\
    \Omega_F \ \times \ \Omega_F\qquad   & \overset {Res}{\longrightarrow} & \quad \ \ \OO_S,
  \end{array}
\end{equation} 
where the $Res$ in the second row is given by the residue symbol (\cite{Hertling})
$$ {\rm Res}\Big[\frac{
  \phi_1(\underline{x},z) \phi_2(\underline{x},z) dx_1\wedge\cdots \wedge
  dx_{n+1}}{\partial_{x_1}F, \cdots, \partial_{x_{n+1}}F} \Big]  ,
$$
for  $\zeta_i=[\phi_i(\underline{x},z) dx_1\wedge\cdots \wedge dx_{n+1}] \in \HH_F^{(0)}$
($i=1,2$).
\end{Proposition} 

\bigskip
\noindent
{\it Note.} 1.  The explicit formula \eqref{eq:Gauss-Manin} implies that the leading terms of the covariant actions of $\nabla_v$ for $v\in\TT_S$ and of $\nabla_{z\partial_z}$ with respect to the filtration on $\HH_F$  are given by the following commutative diagrams, respectively.
\be
\label{eq:initial}
\begin{array} {rclcrcl}
  \HH_F^{(k-1)}\!\! & \overset{\nabla_v}{-\!\!\!-\!\!\! \longrightarrow} & \HH_F^{(k)} 
  &\qquad\qquad  &   \HH_F^{(k-1)}\!\! & \overset{\nabla_{z\partial_z}}{-\!\!\!-\!\!\! \longrightarrow} & \HH_F^{(k)} a
  \\
  \\
 r_F^{(k-1)}\downarrow\ \ \ &  &\ \  \downarrow r_F^{(k)}    
 &\qquad &   r_F^{(k-1)}\downarrow\ \ \ &  &\ \  \downarrow r_F^{(k)}    
\\
 \\
 \Omega_F\ & \overset{\tilde{v}F|_{C_F}}{-\!\!\!-\!\!\! \longrightarrow} &\  \Omega_F  
   &\qquad &  
    \Omega_F\ & \overset{-F|_{C_F}}{-\!\!\!-\!\!\! \longrightarrow} &\  \Omega_F  
 \end{array}
 \ee
 
2.  The following elementary fact shall be used later in the uniqueness of the evolution map $e^{(f-F)/z}$ in \S4. 
The restriction $\HH_F\otimes_{\OO_S}\OO_S/m$ ($m$: the maximal ideal of $\OO_{S,0}$) of $\HH_F$ at the origin $0\in S$ is naturally isomorphic to $\HH_f$ since the homomorphism of $\HH_F\to \HH_F\otimes_{\OO_S}\OO_S/m$ is naturally induced from the pull-back morphism $(\Omega_{Z/S}^*((z)),D_F)
\to (\Omega_{X}^*((z)),D_f)$ of the complexes induced from the inclusion
morphism $X\subset Z$.
\medskip

3.  For a later use in \S6, we give a sketch of a proof of
2. iv). We prove only that the filter $\HH_F^{(0)}$ is preserved by
$\nabla_{z\partial_z+E}$.
 Let $\tilde{E}$ be a lifting of $E$. Then by Definition
\ref{UniversalUnfolding}, there exists $g_i\in \Gamma(Z,\OO_Z)\ \
(i=1,\ldots,n)$ such that $F-\tilde{E}F= \sum_i g_i\partial_{x_i}F$. Therefore,
for any $[s]\in \HH_F^{(0)}$, we calculate
\begin{multline} \nabla_{z\partial_z+E}[s]=\big[(z\partial_z-F/z+\tilde{E}
  +\tilde{E}F/z) s\big] 
  \\ =  \big[(z\partial_z+\tilde{E} -(\sum_i g_i\partial_i F)/z) s\big] =
  \big[(z\partial_z+\tilde{E})s +(d_{Z/S}-z^{-1}D_F)(s')\big] = \\
  =\big[(z\partial_z+\tilde{E})s +d_{Z/S}s'\big] \in \HH_F^{(0)}
\end{multline}
where $s'= (s/(dt_1\wedge\cdots\wedge t_{n+1}))\sum_i(-1)^ig_idt_1\wedge \cdots
\check{dt_i} \cdots \wedge dt_{n+1} \in p_*(\Omega_{Z/S}^{n})$.

\begin{remark}
  \label{Oscillatory} Let us give an analytic interpretation of the modules
  \eqref{eq:FilteredDeRham1} and \eqref{eq:FilteredDeRham2}. 
  
  Choose  $\al\in \Omega_{X}^{n+1}\{z\} \subset \Omega_X^{n+1}[[z]]$ and consider a $z$-dependent
  family of oscillatory integrals for $z \in \CC^*$:
  \begin{equation}
    \int_{\Gamma_z} \ e^{f/z}\ \al
  \end{equation}
  for a continuous family $\{\Gamma_z\}$ of relative cycles $\Gamma_z\in {\rm H}_{n+1}(X, Re(f/z)\!<\!\!<\!
  0)$.  If $\al=D_{z}(\beta)$ for $\beta\in \Omega_{X}^{n}\{z\}$ then 
  \begin{equation}
    \int_{\Gamma_z} \ e^{f/z}\ \al= \int_{\Gamma_z} \ e^{f/z} (zd\beta+df\wedge \beta)
    =z \int_{\Gamma_z} \ d(e^{f/z}\beta) =0,
  \end{equation}
  and we obtain a perfect pairing:
  \begin{equation}
    {\rm H}_{n+1}(X, Re(f/z)\!<\!\!<\! 0) \ \times \
    \Omega_X^{n+1}/D_{z}\Omega_X^{n} \ \rightarrow \ \CC
  \end{equation}
  (\cite{AVG}).  In the present paper, we do not use such analytic
  theory, however we may regard the splitting factor $\hat{\BB}[[z]]$ in \eqref{eq:V((z))}
  as the space of representatives of integrants for ``formal oscillatory
  integrals'' $ \int_{\Gamma_z} \ e^{f/z}\ \al$ out of their equivalence classes. 
\end{remark}

\begin{remark} The notation $\HH_F$ and $\HH^{(k)}_F$ were already used in~\cite{HigherResidue83}
for slightly different objects. Therefore, let us clarify the relationship
between them and \eqref{eq:FilteredDeRham1} and \eqref{eq:FilteredDeRham2}.

 First, let us give a heuristic interpretation of the variable $z$ and the complex
$D_F$. We regard $F$ as a family of functions parametrized by $S$ so that we
have a fibration $\varphi=(F,p): Z\to \mathbf{A}^1\times S$.  Then the equation $D_F=0$ together with the relation \eqref{eq:UnitEuler} should ``imply'' $z= - \big(\frac{d_{Z/S}}{dF}\big)^{-1} \equiv \frac{\partial}{\partial t_1} \bmod dF$.

Therefore, in~\cite{S1,S2}, the variable $z$ was denoted by $\delta_1^{-1}$ in order to
identify it with the pseudo differential operator $\partial_{t_1}^{-1}$ where
$t_1$ the coordinate for the range of $F$ so that the module is defined on the
relative-cotangent space of the range space of $F$. In the present paper, we
regard it simply a formal variable.

\end{remark}

The modules $\HH_f$ and $\HH^{(k)}_f$ were sometimes called the module of
semi-infinite Hodge structure for the vanishing cycles of $f$. However, in the
present paper, we regard them as the spaces of integrants for oscillatory integrands with respect to the exponential
factor $\exp(f/z)$ (see
Remark\ref{Oscillatory}). It is of our main interests in \S4 how they evolve to the spaces of
 integrants for the oscillatory integrals with respect to the exponential
factor $\exp(F/z)$ for an unfolding $F$ of $f$.

\medskip
\section{Section and opposite filtration of $\HH_F$}

Let the setting be as in \S2, i.e.\ $F$ is an unfolding define on $Z\subset X\times S$ of a function $f$ having an isolated critical point at  $0\in X\! \subset\! \CC^{n+1}$, where $S$ is the parameter space for the unfolding. We analyze the filtered $\OO_S$-module $\HH_F$ by introducing concepts of
sections and opposite filtrations (see Remark 3.1.).

\begin{definition} 1.  A $\OO_S$-submodule $\mathcal{B}$ of $\HH_F^{(0)}$ is
called a \textit{section} of $\HH_F^{(0)}$, if the restriction of the projection $r_F^{(0)}$ (recall 
\eqref{eq:Filter2}) induces  an $\OO_S$-isomorphism  $\mathcal{B}\simeq \Omega_F$. In view of Proposition
\ref{GMHR}, the condition on $\mathcal{B}$ is equivalent to that we have natural $\OO_S[[z]]$-isomorphisms
\be
\label{eq:Section} \BB[[z]] \ \simeq \ \HH_F^{(0)} \qquad \text{and} \qquad
\BB((z)) \ \simeq \ \HH_F .  \ee The set of all sections of $\HH_F^{(0)}$ is
denoted by $\MM_{s,F}$.

2. A $\OO_S$-submodule $\mathcal{L}$ of $\HH_F$ is called an opposite filtration
of $\HH_F$, if it satisfies (1) the $\OO_S$-splitting condition:
$\HH_F=\HH_F^{(0)}\oplus \mathcal{L}$ and (2) the opposite condition:
$z\LL\supset \LL$. The set of all opposite filtrations of $\HH_F$ is denoted by
$\MM_{o,F}$.
\end{definition}

For a later application in the proof of Theorem 6.1, below we describe the isomorphism
\eqref{eq:Section} for the case of  the trivial unfolding down to earth.  

Consider any
$\mu$-dimensional $\CC$-subspace $\hat{\BB}$ of $\Gamma(X,\Omega_X^{n+1})[[z]]$ such that it is
isomorphic to $\Omega_f$ by the projection $\Omega_X^{n+1}[[z]] \bmod \big(dF\wedge\Omega_X^n[[z]]+ z\Omega_X^{n+1} [[z]]\big)  \to \Omega_f$. 
In particular, the condition implies 
\be
\label{eq:linearIndependence}
\hat{\BB} \  \cap \ z \Gamma(X,\Omega^{n+1}_X)[[z]] \ = \ 0.
\ee
I.e.\ the projection of $\hat{\BB}$ into $\Gamma(X,\Omega^{n+1}_X)$ is an embedding.
This implies that we get a sequence of injective morphisms: $ \hat{\BB}[z]/(z^i) \to \Gamma(X,\Omega^{n+1}_X)[z]/(z^i)$ for $i=1,2, 3,\cdots$ (since a graded piece of the morphisms is equivalent to the embedding for $i=1$). Taking the inverse limit of the sequence, we see that  {\it $\hat{\BB}[[z]]$ is embedded into $\Gamma(X,\Omega^{n+1}_X)[[z]]$}. 
The condition on $ \hat{\BB}$ implies further the splittting:
\be
\label{eq:split-1}
\Gamma(X,\Omega^{n+1}_X) \quad = \quad \text{the projection image of }\hat{\BB} \text{ in } \Gamma(X,\Omega^{n+1}_X) \ \oplus \ \
\Gamma(X,df\wedge \Omega_X^{n}).
\ee
Then, repeating the similar arguments above, we can show the following splitting.
\be
\label{eq:split0} 
\Gamma(X,\Omega^{n+1}_X)[[z]] \quad = \quad \hat{\BB}[[z]] \ \oplus \ \
\Gamma(X,df\wedge \Omega_X^{n})[[z]].
\ee 
Note also that the modules
$\hat{\BB}[[z]]$ and $\hat{\BB}((z))$ are equal to $\hat{\BB}\otimes_{\CC}\CC[[z]]$ and $\hat{\BB}\otimes_{\CC}\CC((z))$, respectively, due to the fact $\dim_\CC \hat{\BB}=\mu <\infty$, and are free of rank $\mu$ over $\CC[[z]]$ and $\CC((z))$, respectively.

\begin{Proposition} 
\label{eq:split1} 
Let $ \hat{\BB}\subset \Gamma(X,\Omega_X^{n+1})[[z]]$ be as above. Then, there are splittings 
\begin{equation}
  \begin{array}{ccccc}
    \label{eq:V((z))} 
    &\Gamma(X,\Omega^{n+1}_X)((z)) \qquad &=& \quad \hat{\BB}((z)) \
                                              \oplus \ D_f\big(\Gamma(X,\Omega^{n}_X)((z))\big) 
    \\
    &z^{-k}\Gamma(X,\Omega^{n+1}_X)[[z]]\quad \qquad &=& \quad z^{-k}\hat{\BB}[[z]] \
                                                         \oplus \ D_f\big(z^{-k}\Gamma(X,\Omega^{n}_X)[[z]]\big) .
  \end{array}
\end{equation}
\end{Proposition}
\begin{proof}
  It is sufficient to show the second formula. 
  
  \smallskip
  1.  Since each factor in the RHS of \eqref{eq:V((z))} is embedded in the LHS, we consider the natural homomorphism from the RHS to the LHS.
  We first show that it is surjective. 
  
  Pick an element $\alpha$ from the $k$th filter $z^{-k}
  \Gamma(X,\Omega^{n+1}_X)[[z]]$.
  Due to \eqref{eq:split0}, we find $v^1\in z^{-k} \hat{\BB}[[z]]$ and $\beta^1 \in z^{-k} \Gamma(X,\Omega_X^{n})[[z]]$ such that


$$
     \alpha = \dd f \wedge \beta^1 + v^1. 
$$
 Setting  $\alpha^1:=\alpha$ and $\alpha^2:= -d\beta^1$,  we obtain
 $$
    \alpha^1 = v^1 + D_f \beta^1 + z \alpha^2,
$$
Since $\alpha^2\in z^{-k} \Gamma(X,\Omega^{n+1}_X)[[z]]$, we continue inductively    
 $$
 \begin{array}{cc}
    &\alpha^2 = v^2 + D_f \beta^2 + z \alpha^3, \\
    \\
    &\alpha^3 = v^3  + D_f \beta^3 + z \alpha^4, \\
    &  \qquad \ldots
    \end{array}
 $$
    where $\alpha^i := -\dd \beta^{i-1} \in z^{-k} \Gamma(X,\Omega^{n+1}_X)[[z]]$ for $i\ge2$.
 Multiplying the $i$-th line by $z^{i-1}$ and summing up all lines, we obtain the decomposition:
 \begin{equation} \label{dec1}
   \alpha =  v  + D_f \beta.
 \end{equation}
where $v=\sum_{i=1}^\infty v^i z^{i-1}\in z^{-k}\hat{\BB}[[z]]$ and
$\beta=\sum_{i=1}^\infty \beta^i z^{i-1}\in z^{-k}\Omega_X^{n}[[z]]$.  \quad $\Box$

\medskip

2.  Next, let us show that $ \hat{\BB}[[z]] \cap
D_f\big(\Gamma(X,\Omega_X^{n})[[z]]\big) =\{0\}$.

Suppose the contrary, and there exist $v=\sum_{i=0}^\infty v_i z^i
\in \hat{\BB}[[z]]$ ($v_i\in \hat{\BB}$), $ v_i=\sum_{j=0}^\infty v_{ij}z^j$ ($v_{ij}\in \Gamma(X,\Omega_X^{n+1})$)  and $\al=\sum_{i=0}^\infty \alpha_i z^i \in 
\Gamma(X,\Omega^{n}_X)[[z]]$ ($\al_i\in \Gamma(X,\Omega^{n}_X)$)  such that $v=D_f\al$.  Let us show by induction
on $i\in \ZZ_{\ge 0}$ the following assertion.

\medskip
\noindent $(*)_{i}$: the vanishing $v_i=0$ and existences of $\beta_i \in
\Gamma(X,\Omega_X^{n-1})$ such that $\al_i=df\wedge \beta_i+d\beta_{i-1}$.

\medskip 
Let us write the equality $v=D_f\al$ term-wisely: $\sum_{k+l=i} v_{kl}= d\al_{i-1} +df\wedge
\al_i$ for $i\in \ZZ_{\ge 0}$, where $\al_{-1}=0$. The equality for the case $i=0$ turns to be $v_{00}=df\wedge \alpha_0$. Then, in view of the
splitting \eqref{eq:split-1}, we see $v_{00}=0$ and $df\wedge \alpha_0=0$. Then, \eqref{eq:linearIndependence} implies $v_0=0$, and 
applying Poincare Lemma, we see $\alpha_0=df\wedge \beta_0$ for $\beta_0\in
\Gamma(X,\Omega_X^{n-1})$. That is $(*)_0$ is shown by setting
$\beta_{-1}=0$. Suppose that $(*)_i$ is shown for $i\in \ZZ_{\ge 0}$. Then we
have
$$
v_{i+1,0}=d(df\wedge\beta_i+d\beta_{i-1}) + df\wedge
\al_{i+1}=df\wedge(-d\beta_i+\al_{i+1}).
$$
Applying again the splitting \eqref{eq:split-1} to this equality, we get $v_{i+1,0}=0$ and hence  $v_{i+1}=0$,  and get the
existence of $\beta_{i+1}\in \Gamma(X,\Omega_X^{n-1})$ such that $\al_{i+1}=df\wedge \beta_{i+1}+ d\beta_i$,
implying $(*)_{i+1}$. The induction shows that $v=0$ and $\al= D_f\beta$.
\quad $\Box$
    
\smallskip    
   Above 1. and 2. together  give the proof of the splittings \eqref{eq:V((z))}.
\end{proof}


\smallskip In later sections, we shall abuse the notion of sections and opposite
filtrations for the cases when the coefficient is no-longer $\OO_S$ but is
either an Artinian ring $\OO_{S,0}/m^{k+1}$ or the formal rings
$\check{\OO}_{S,0}$.

\begin{remark}
\label{WithOrWithout} Present definitions of section and opposite filtration and
the coming Definition 6.1 of good pair are ``weaker" copies of those Definition
2.16 and Definition 4.13 in~\cite{LLS} in the sense that we do not assume
for them the metric condition discussed in the following Proposition
\ref{Metric}. except we explicitly require it, whereas in~\cite{LLS} the metric
condition is assumed for every section, opposite filtration and good
pair. Nevertheless, we shall see that the most of statements in~\cite{LLS} are recovered.
\end{remark}

\medskip
\begin{Proposition} 
\label{SectOppo} There is a natural bijection: 
\be
\label{eq:Opposite-Section} \MM_{o,F} \ \simeq \ \MM_{s,F} , \quad \LL \ \mapsto
\ \BB(\LL):=\HH_F^{(0)}\cap z\LL .  \ee The inverse correspondence is given by
$\BB\mapsto \LL(\BB):=z^{-1}\BB[z^{-1}]$.
\end{Proposition}
\begin{proof} (c.f.~\cite{LLS} 2.16) Let $\LL\in \MM_{o,F}$ be given. Multiplying
 the decomposition in Definition 3.1.\ 2.\ (1) by $z^k$ for $k\in \ZZ$, we obtain
$\HH_F=\HH_F^{(-k)}\oplus z^k\mathcal{L}$, and hence
$\HH_F^{(-k)}=\HH_F^{(-k-1)}\oplus (\HH_F^{(-k)}\cap z^{k+1} \mathcal{L})$. For
$k=0$, we obtain $\HH_F^{(0)}=\HH_F^{(-1)}\oplus \mathcal{B}(\mathcal{L})$,
implying the isomorphism $\mathcal{B}(\mathcal{L})\simeq \Omega_F$ due to
\eqref{eq:Filter2} for $k=0$.

Conversely, let $\mathcal{B}\in \mathcal{M}_{s,F}$ be given. Then, applying
\eqref{eq:Filter1} and \eqref{eq:Filter2} repeatedly, we obtain
$\HH_F^{(k)}=\HH_F^{(0)} \oplus\bigoplus_{i=1}^k z^{-i}\BB$ for all $k\in
\ZZ_{>0}$. In view of \eqref{eq:Filter1'}, we see that
$\LL(\BB):=z^{-1}\BB[z^{-1}] \in \mathcal{M}_{o,F}$.

By the construction and the property Definition 3.1.\ 2.\ (2) of $\LL$, it is
clear that $\LL(\BB(\LL))\subset \LL$. Since we have the decompositions
$\HH_F=\HH_F^{(0)}\oplus \LL= \HH_F^{(0)}\oplus \LL(\BB(\LL))$, we get
$\LL(\BB(\LL))= \LL$.

Conversely, by the construction, it is clear that $\BB(\LL(\BB))\supset \BB$.
Since we have the decompositions $\HH_F^{(0)}=\HH_F^{(-1)}\oplus \BB=
\HH_F^{(-1)}\oplus \BB(\LL(\BB))$, we get $\BB(\LL(\BB))= \BB$.
\end{proof}

\begin{Proposition}
\label{z-stable} Let $\BB\in \MM_{s,F}$ and $\LL\in \MM_{o,F}$ correspond to
each other by \eqref{eq:Opposite-Section}. Then the following conditions a), a')
and b) are equivalent. We shall call these conditions the
$z\partial_z$-stability condition.

a) $\nabla_{z\partial_z}\BB \subset \BB+z^{-1}\BB$ , a')
$\nabla_{z\partial_z}\BB \subset \BB[z^{-1}]$ , b) $\nabla_{z\partial_z}\LL
\subset \LL$.
\end{Proposition}
\begin{proof} Clearly, a) implies a').

a')$\Rightarrow$b): $\nabla_{z\partial_z}\LL=\nabla_{z\partial_z}
z^{-1}\BB[z^{-1}]\subset (\nabla_{z\partial_z} \BB)z^{-1}\CC[z^{-1}] + \BB
z\partial_z(z^{-1}\CC[z^{-1}]) \subset \BB z^{-1}\CC[z^{-1}]) =\LL$.

b)$\Rightarrow$a): $\nabla_{z\partial_z}\BB=\nabla_{z\partial_z}
(\HH_F^{(0)}\cap z\LL) \subset \nabla_{z\partial_z} (\HH_F^{(0)}) \cap
\nabla_{z\partial_z} (z\LL) \subset z^{-1}\HH_F^{(0)} \cap z\LL =\BB +
z^{-1}\BB.$
\end{proof}

\begin{Proposition}
\label{Metric} Let $\BB\in \MM_{s,F}$ and $\LL\in \MM_{o,F}$ correspond to each
other by \eqref{eq:Opposite-Section}. Then the following two conditions c) and
d) are equivalent. We shall call these conditions the metric condition.

c) The set of values $K_F(\BB,\BB):=\{K_F(v_1,v_2)\mid v_1,v_2\in\BB\}$ is
contained in $z^n \OO_S $ \ (c.f.\ \eqref{eq:HRdeg}).

d) The set of values $K_F(\LL,\LL):=\{K_F(l_1,l_2)\mid l_1,l_2\in\LL\}$ is
contained in $z^{-2+n}\OO_S [z^{-1}]$.
\end{Proposition}
\begin{proof} c)$\Rightarrow$ d) : $K_F(\LL,\LL)\subset
K_F(z^{-1}\BB[z^{-1}],z^{-1}\BB[z^{-1}])\subset
K_F(\BB,\BB)z^{-2}\OO_S[z^{-1}]\subset z^{-2+n}\OO_S[z^{-1}]$.

d)$\Rightarrow$ c) : $K_F(\BB,\BB)\subset K_F(\HH_F^{(0)}\cap
z\LL,\HH_F^{(0)}\cap z\LL)\subset K_F(\HH_F^{(0)},\HH_F^{(0)}) \cap
K_F(z\LL,z\LL) \subset z^{n}\OO_S[[z]] \cap z^{n}\OO_S[z^{-1}]=\OO_S\cdot
z^{n}$.
\end{proof}


\begin{remark}
\label{simplectic} According to Givental \cite{G}, we may regard $\LL$ as a
$\OO_S$-linear Lagrangean subspace by setting the symplectic form
$$
\omega :\HH_F \times \HH_F \to \OO_S, \quad \omega(s_1,s_2):=
\mathrm{Res}_{z=0}\big(z^{-n} K_F(s_1,s_2)\big) .
$$
Then, the initial property \eqref{eq:HRdeg} implies
$\omega(\HH_F^{(0)},\HH_F^{(0)})=0$, and the metric condition in Proposition
\ref{Metric} is equivalent to $\omega(\LL,\LL)=0$. That is,
$\HH_F=\HH_F^{(0)}\oplus \mathcal{L}$ is a {\it holonomic decomposition} in the
sense of Mikio Sato.
\end{remark}

\section{Formal analysis on deformation parameter space} \label{sec:Formal}

From this paragraph, we take a perturbative approach with respect to the
deformation variable $\underline{t}$ in $S$ at $0$, as was done in~\cite{LLS}.
Namely, we embed $\HH_F$ into its formal completion $\check{\HH}_F$ with
respect to the deformation variable $\underline{t}$.  We show that the formal
(infinite) bundle $\check{\HH}_F$ over the formal space $\rm{Spf}\
\check{\OO}_{S,0}$ is trivialized (w.r.t.\ the Gauss-Manin connection $\nabla$)
by the evolution of $\HH_f$ by the oscillatory integral factor $\exp{(f\!-\!F)/z}$ (Proposition 4.1).
In particular, sections $\BB$ in $\HH_f^{(0)}$ and opposite filtrations $\LL$ in
$\HH_f$ evolve to flat sections and flat opposite filtrations in
$\check{\HH}_F$, respectively. This fact leads to the construction of
formal primitive forms (without metric) in the next section.

\medskip Consider any unfolding $F$ defined over the frame $(Z,S,p,F)$ in 2.1.
Let $0\in S$ be the base point and let $m$ be the maximal ideal of the local ring $\OO_{S,0}$. Then  $m^{k+1}$ for $k\in \ZZ_{\ge0}$ form the $k$-th formal neighborhood of $0\in \OO_{S,0}$ with respect to the $m$-adic topology on  $\OO_{S,0}$. Set
$\check{\OO}_{S,0}:=\underset{\underset{k}{\leftarrow}}{\lim} \
\OO_{S,0}/m^{k+1}$ and $\check{\mathcal{F}}:=\mathcal{F}\hat{\otimes}_{\OO_S} \!
\check{\OO}_{S,0}:=\underset{\underset{k}{\leftarrow}}{\lim}\
\mathcal{F}\otimes_{\OO_S} \! \OO_{S,0}/m^{k+1}$
for a $\OO_S$-module $\mathcal{F}$.  Using a local coordinate system
$\underline{t}=(t_1,\cdots,t_m)$ of $S$ at the origin 0, we
sometimes denote $\check{\OO}_{S,0}$ also by $\CC[[\underline{t}]]$ and
$\check{\mathcal{F}}$ by $\mathcal{F}[[\underline{t}]]$.

Recall Proposition \ref{GMHR} that the $\OO_S$-module $\HH_F$ admits the
Gauss-Manin connection $\nabla$ and the higher residue pairings $K_F$. We extend
them to $\check{\HH}_F :=\underset{\leftarrow}{\lim}\
\mathcal{H}_F\otimes_{\OO_S} \! \OO_{S}/m^{k+1}$ as follows. Namely, one easily
shows (see \cite{LLS}, Lemma 4.3) that $\nabla$ induces a connection 
\be
\label{eq:FormalNabla} \check{\nabla} \ :\ \mathcal{T}_{\mathbf{A}_z^1\times S}
\times \check{\HH}_F \ \longrightarrow \ \check{\HH}_F 
\ee 
defined by the
projective limit of $\mathcal{T}_S \times \HH_F/m^{k+1}\HH_F \to \HH_F/m^k\HH_F$
and $\mathcal{T}_{\mathbf{A}_z^1} \times \HH_F/m^{k+1}\HH_F \to
\HH_F/m^{k+1}\HH_F$ for $k\in \ZZ_{>0}$, which we call again the Gauss-Manin
connection and denote again by $\nabla$.  The higher residue pairings $K_F$ on
$\HH_F$ also induces an $(\OO_S((z)))^\vee$-skew-symmetric pairing: \be
\label{eq:FormalHR} \check{K}_F: \check{\HH}_F \times \check{\HH}_F \
\longrightarrow \ (\OO_S((z)))^\vee.  \ee They satisfy the properties parallel
to those listed in 2. and 3. of Proposition \ref{GMHR}, whose detailed
formulations and verifications are left to the reader.

\begin{remark} Let us consider the formal unfolding $\check F$ (associated with an unfolding $F$) defined on the
formal frame $(\check{Z},\check{S}, p,\check{F})$ (that is,
$\check{S}=\mathrm{Spf} (\check{\OO}_{S,0})$ is the formal neighborhood of the
origin $0\in S$, $\check{Z}$ is a formal neighborhood of $X\times\{0\}$ in
$\CC^n\times \check{S}$ and $\check{F}$ is the formal power series expansion of $F$ with
respect to a local coordinate system $\underline{t}$ of $S$ at $0$. Detailed
descriptions are left to the reader).  Then, analogous to \S2, one can define
the $\check{\OO}_S((z))$-module
$\HH_{\check{F}}:=p_*(\Omega^n_{\check{Z}/\check{S}}((z))/D_{\check{F}}
(p_*(\Omega^{n-1}_{\check{Z}/\check{S}}((z))))$ which admits increasing
filtration $\HH^{(k)}_{\check{F}}$ satisfying the properties parallel to
\eqref{eq:FilteredDeRham1},\eqref{eq:FilteredDeRham2},\eqref{eq:Filter1} and
\eqref{eq:Filter2} (here, we note that
$\Omega_{\check{F}}=\check{\Omega}_F=\Omega_F\otimes_{\OO_S}\check{\OO}_{S,0}$).
In particular, we have isomorphisms $\HH_{\check{F}} \simeq
\HH_F\otimes_{\OO_S((z))}\check{\OO}_{S,0}((z))$, which admits the filtation
$\HH_{\check{F}}^{(k)} \simeq
\HH_F^{(k)}\otimes_{\OO_S[[z]]}\check{\OO}_{S,0}[[z]] \ (= \check{\HH}_F^{(k)})$
for $k\in\ZZ$ (regarded as the formal version of $\HH_F$).

However, due to the ``non-commutativity"
$\check{\OO}_{S,0}((z))=\CC[[\underline{t}]]((z)) \underset{\not=}{\subset}
\OO_{S,0}((z))^\vee=\CC((z))[[\underline{t}]]$ (since the RHS is a formal power
series in $(\underline{t})$ of coefficients in Laurent series of $z$ whose order
of poles may not necessarily bounded, e.g.\ $\sum_{k=0}^\infty
z^{-k}t^k=z/(z-t)=-\sum_{m=1}^\infty t^{-m}z^m$ is in the RHS but not in the
LHS), we have \be
\label{eq:inclusion} \HH_{F,0} \ \underset{\not=}{\subset} \ \HH_{\check{F}} \
\underset{\not=}{\subset} \ \check{\HH}_F \ .  \ee So, even if the module
$\check{\HH}_F$ has the increasing sequence of
$\CC[[z]][[\underline{t}]]=\CC[[z,\underline{t}]]$-submodules
$\check{\HH}^{(k)}_F\simeq \HH^{(k)}_{\check{F}}$ ($k\in\ZZ$), the union
$\cup_{k\in \ZZ} \check{\HH}^{(k)}_F$ covers only the submodule
$\HH_{\check{F}}$ consisting of elements which admit bounded Laurent series expansion in
$z$ with coefficients in $\CC[[\underline{t}]]$ (w.r.t.\ a fixed basis), but does
not cover the full-module $\check{\HH}_F$ which consists of power series in
$\underline{t}$ whose coefficient in a monomial in $\underline{t}$ is a Laurent
series in $z$ but the order of poles in $z$ may grow to infinity as the order of
the monomial increases. Thus, the analogy of the regularity \eqref{eq:Filter1'} does
not hold for $\check{\HH}_F$.
\end{remark}

\medskip We now compare the given unfolding $F$ over a frame $(Z,S,p)$ with the trivial unfolding
$f$. Precisely, recall the projection $\pi_X:Z\to X$ (\S2), and depending on
that, we define the trivial unfolding
$$
F_0\quad :=\quad \pi_X^{-1}(f)
$$ 
over the frame $(Z,S,p)$ so that $F-F_0\in m \OO_Z$.  Then following Propositions 4.1 and 4.2,
giving a trivialization of $\check{\HH}_F$ and an opposite filtration of $\check{\HH}_F$, are the
key results in the formal analysis in the present paper.

\begin{Proposition}
\label{FlatExtension1} \rm{ (\cite{LLS}, Lemma/Definition 4.7) } 1.  Consider a
correspondence \be
\label{eq:exponential} e^{(f-F)/z} \ : \ \HH_f \ \longrightarrow \check{\HH}_F.
\ee given by associating, for each $[s] \in \HH_f$, the projective system of
elements
\be
\label{eq:Exponential} \Big[ \sum_{0\le l\le k}\frac{(F_0-F)^l}{z^l l!}
\pi_X^{-1}(s)\Big] \ \in \ \HH_F \bmod m^{k+1}\HH_F \quad \text{for } k\in
\ZZ_{>0} , \ee
where $s\in p_{0*}(\Omega_{X}^n((z)))$ is any representative of $[s]$.  Then, it
defines a well-defined $\CC((z))$-module injective homomorphism, which we shall call the
evolution map.\footnote
{We use the notation  $e^{(f-F)/z}$ instead of  $e^{(F_0-F)/z}$, since we shall see immediately in Corollary 4.1 that the correspondence does not depend on the choice of the projection $\pi_X$ and, hence, of $F_0$.
}

\medskip
2. The composition of the evolution map $e^{(f-F)/z}$ \eqref{eq:exponential} with the restriction map 
$\check{\HH}_F\to \check{\HH}_F\otimes _{\check{\OO}_{S,0}}(\check{\OO}_{S,0}/m)\simeq \HH_f$ (c.f.\ \S2 Note. 2.) is the identity map on $\HH_f$.

\medskip
3.  The image $e^{(f-F)/z}[s]\in  \check{\HH}_F$ for $[s]\in \HH_f$ of the evolution
map \eqref{eq:exponential} is horizontal (flat) with respect to the
Gauss-Manin connection on $\check{\HH}_F$ over $\TT_S$-direction (but not
$\TT_{\mathbf{A}_z^1}$-direction). Precisely, 
\be
\label{eq:flatness}
\ker\big( \nabla_{\TT_S} \big)\ = \ e^{(f-F)/z} \big( \HH_f \big) .  
\ee
 

 \medskip
4. The evolution map \eqref{eq:exponential} induces an
$\CC((z))\hat{\otimes}_{\CC}\check{\OO}_{S,0}:=\underset{\underset{k}{\leftarrow}}{\lim}\
\CC((z)) \! \otimes_{\CC}\! \OO_{S,0}/m^{k+1}$-isomorphism: \be
\label{eq:evolve}
\begin{array}{c} \HH_f \ \hat{\otimes}_{\CC} \ \check{\OO}_{S,0} \quad \simeq
\quad \check{\HH}_F \\ \HH_f\otimes_\CC \OO_S/m^{k+1}
\overset{e^{(F_0-F)/z}}{-\!\!\!-\!\!\!-\!\!\!-\!\!\!-\!\!\!-\!\!\!\longrightarrow}
\HH_F/m^{k+1}\HH_F, \quad [s]\otimes [r_k] \mapsto e^{(F_0-F)/z} \pi_X^{-1}(s)
r_k \bmod m^{k+1}\HH_F.
\end{array} \ee

The inverse map of  \eqref{eq:evolve} 
is given by multiplying $e^{(F-F_0)/z} \bmod m^{k+1}$ at each $k$-th
neighborhood, and is denoted by $e^{(F-f)/z}$. We shall call the extended
isomorphism \eqref{eq:evolve} also the evolution map.

\medskip
5.   (i) The evolution isomorphsim \eqref{eq:evolve} is equivariant with the Gauss-Manin connection $\nabla_{z\partial_z}$ \eqref{eq:Gauss-Manin}, where its action in LHS is defined to be commutative with the action of $\check{\OO}_{S,0}$. 

 (ii)  The evolution isomorphsim \eqref{eq:evolve} is equivariant with the action of $\TT_S$, where $v\in \TT_S$ acts on LHS as the derivation on the coefficient $\check{\OO}_{S,0} $ and on RHS via Gauss-Manin connection $\nabla_v$ \eqref{eq:Gauss-Manin}.

\end{Proposition}
\begin{proof} The proof given below is inspired from the oscillatory integral
view point (see Remark \ref{Oscillatory}).

1. First, remark that the $k$-th neighborhood $\HH_F\otimes_{\OO_S}
\OO_{S}/m^{k+1}= \HH_F \bmod m^{k+1}\HH_F$ in the RHS of \eqref{eq:Exponential}
has an expression 
\be
\label{eq:Infinitesimal} \Big(p_*(\Omega^{n+1}_{Z/S}((z))) \bmod
m^{k+1}p_*(\Omega^{n+1}_{Z/S}((z)))\Big)\Big/\Big(D_F(p_*(\Omega^{n}_{Z/S}((z)))\bmod
m^{k+1}D_F(p_*(\Omega^{n}_{Z/S}((z))\Big) 
\ee 
that is, quotient out by the co-boundary operator $D_F$ and quotient out  by the higher terms $m^{k+1} $ are commutative to each other (this fact can be seen by the right-exactness of
the tensor product $\otimes _{\OO_S}\OO_{S}/m^{k+1}$). Then, in order to show
the well-definedness of the map \eqref{eq:exponential}, we need to show the
following two facts.

1) For a given class $[s] \ \in \ \HH_f$, fix its representative $s\in
p_{0*}(\Omega_{X}^n((z)))$.  Since the order of the pole in $z$ of the
expression in the bracket of \eqref{eq:Exponential} for each $k\in \ZZ_{>0}$ is
bounded by $k+$the order of the pole in $z$ of $s$, the expression defines an
element in the numerator of \eqref{eq:Infinitesimal}. Further more, the sequence
of the elements for $k\in\ZZ_{>0}$ forms a projective system, since we have
$e^{(F_0-F)/z} =\sum_{l=0}^\infty \frac{(F_0-F)^l}{z^l l!}$ such that $
\frac{(F_0-F)^l}{z^l l!} \in m^{l} p_*\OO_Z((z))$ (recall Definition 2.1, 2. so
that $F_0-F\in m \cdot p_*\OO_Z$).

2) The class of the image in \eqref{eq:Exponential} depends not on the choice of
the representative $s\in p_{0*}(\Omega_{X}^n((z)))$ but only on the class
$[s]\in \HH_f$. ({\it Proof.} Consider $s'$ such that $s'-s=D_f(r)$ for $r\in
p_{0*}(\Omega_{X}^{n-1}((z)))$.  Then, in the module of the numerator of
\eqref{eq:Infinitesimal}, we calculate the difference of the representatives of
$e^{(f-F)/z}[s']-e^{(f-F)/z}[s]$: 
\ba 
\label{eq:DfDFcommutative}
e^{(f-F)/z}\pi_X^{-1}(D_{f}(r)) \equiv
e^{(F_0-F)/z}D_{F_0}(\pi_X^{-1}(r)) \equiv D_{F}( e^{(F_0-F)/z}(\pi_X^{-1}(r)) )
\quad \bmod m^{k+1} 
\ea 
where in the RHS, the order of poles in $z$ is again
bounded by $k$ so that the expression defines an element in the denominator of
\eqref{eq:Infinitesimal}, i.e.\ $e^{(f-F)/z}[s']-e^{(f-F)/z}[s] =0$. So the
evolution map is well-defined. It commutes with the action of $\CC((z))$ since
it is so for $\pi_X^{-1}$ and in $k$th neighborhood \eqref{eq:Exponential}.

2.  This is trivial due to the initial term of the explicit expresssion \eqref{eq:Exponential} of the evolution map.

3.  We first show that RHS of \eqref{eq:flatness} is included in LHS. This follows formally by the calculation:
$$
\begin{array}{rll} (\nabla_v (e^{(f-F)/z}[s]) \bmod m^k & \equiv & \nabla_v
(e^{(F_0-F)/z}\pi_X^{-1}(s) \bmod m^{k+1}) \\ & \equiv & \big[
(L_{\tilde{v}}+\tilde{v} F/z)(e^{(F_0-F)/z} \pi_X^{-1}(s))\big] \bmod m^k
\end{array}
$$
where $L_{\tilde{v}}$ means the Lie derivative with respect to a lifting
$\tilde{v}$ of $v$ and the lifting $\tilde{v}$ may be chosen such that
$\tilde{v} (\pi_X^{-1}(\OO_X))=0$. Then, we have $ L_{\tilde{v}}(e^{(F_0-F)/z}
\pi_X^{-1}(s))=(-\tilde{v}F/z)e^{(F_0-F)/z} \pi_X^{-1}(s)$ so that the two terms
in the RHS of the above expression cancel to each other and the sum vanishes.

The opposite inclusion is shown as follows. Let $\zeta$ be any element of LHS. Consider the evolution $e^{(f-F)/z}(\zeta|_0)$ where $\zeta|_0=\zeta\bmod m \in \HH_f$ is the restriction of $\zeta$ at the origin $0\in S$ (recall \S2 Note.2.).
Then the uniqueness of the solution of the linear equation $\nabla_{\TT_S}\zeta=0$ with the same initial value $\zeta|_0$ (recall 2.\ of present Proposition) implies the coincidence of $\zeta$ with $e^{(f-F)/z}(\zeta|_0) \in$ RHS.


4. The evolution map $\HH_f \to \check{\HH}_F$ \eqref{eq:exponential} is a
$\CC((z))$-morphism, and the correspondence \eqref{eq:evolve} naturally extends
the map to a $\big(\OO_{S,0}((z))\big)^\vee$-morphism $\HH_f\hat{\otimes}_\CC
\check{\OO}_{S,0} \to \check{\HH}_F$ between the free modules. Its residue map
modulo the maximal ideal $\CC((T))\check{m}$ is the natural identity $\HH_f
\overset{\sim}{\ \to}\ \HH_{F} |_{\underline{t}=0}=\HH_f$. Then, due to
Nakayama's lemma, the morphism is isomorphic.

It is formal to see that the formal multiplication by $e^{(F-F_0)/z}$ gives the
inversion map to $e^{(f-F)/z}$.

\smallskip
5. (i) In the cochain complex level, the evolution morphism statisfies the following commutativity:
$$
e^{(f-F)/z} \cdot  \big(z\partial_z - f/z \big) \ =\ \big(z\partial_z -F/z) \cdot  e^{(f-F)/z}
$$
which is shown by a direct calculation. Then, in view of the formula \eqref{eq:Gauss-Manin} for $\nabla_{z\partial_z}$, we obtain the commutativity in the cohomological level.

(ii) This follows from the flatness 3.\ of present Proposition and the fact that the evolution morphism is an $\big(\OO_{S,0}((z))\big)^\vee$-linear morphism, since 
$\nabla_v 
e^{(f-F)/z} ([s]\hat{\otimes} a)=\nabla_v  (e^{(f-F)/z} [s]) \hat{\otimes} a
=\nabla_v (e^{(f-F)/z} [s]) \hat{\otimes} a+(e^{(f-F)/z} [s]) \hat{\otimes} va
= (e^{(f-F)/z} [s]) \hat{\otimes} va= e^{(f-F)/z}( [s] \hat{\otimes} va)$ for $[s]\in\HH_f$ and $a\in \big(\OO_{S,0}((z))\big)^\vee$.


This completes a proof of Proposition \ref{FlatExtension1}
\end{proof}

\begin{Corollary}
The evolution map $e^{(f-F)/z}$ \eqref{eq:exponential} is defined independent of a choice of the projection $\pi_X$ and, hence, of $F_0=\pi_X^{-1}f$. This justifies the use of $f$ in the notation $e^{(f-F)/z}$.
\end{Corollary}
\begin{proof}
For any element $[s] \in \HH_f$, the image $e^{(f-F)/z}[s]$ is uniquely characterized by the two conditions: (1) it satisfies the system of linear equations $\nabla_v (e^{(f-F)/z}[s])=0$ for all $v\in \TT_S$ (Proposition \ref{FlatExtension1}, 3.), and (2) its initial value $(e^{(f-F)/z}[s])|_0$ at the origin $0\in S$ is equal to the original $[s]\in \HH_f$ (Proposition \ref{FlatExtension1}) due to the uniqueness of the solution of the system of linear partial differential equation. Obviously, these two conditions do not depend on the choice of the projection $\pi_S$. So the solution $e^{(f-F)/z}[s]$, i.e.\ the image of $[s]$ by the evolution map, does not depend on the choice of $\pi_X$.
\end{proof}


\begin{remark} The Proposition \ref{FlatExtension1} says that the ``infinite
dimensional bundle" $\HH_f$ over a point $\{0\}$ flatly evolves to the trivial
bundle $\check{\HH}_F$ over the infinitesimal neighborhood $\check{S}=Spf
(\check{\OO}_{S,0})$ of $\{0\}$.

 Note that the lattice $\HH_f^{(0)}\subset \HH_f$ does not evolute into the lattice
$\check{\HH}_F^{(0)}\subset \check{\HH}_F$ by the formal evolution map
\eqref{eq:exponential} (since the oscillatory integral factor $e^{(f-F)/z}$ has
pole in $z$). However, as we see in the next Proposition that an opposite
filtration $\LL$ of $\HH_f$ evolves to a formal opposite filtration in
$\check{\HH}_F$.
\end{remark}

\smallskip We now want to apply this formal evolution map to the sections and
opposite filtrations on $\HH_f$.

Consider a section $\BB\in \MM_{s,f}$ and a opposite filtration $\LL\in
\MM_{o,f}$ related as \eqref{eq:Opposite-Section} for the trivial unfolding
$f$. We first formally extend them as ${\BB}\hat{\otimes}_{\CC}
\check{\OO}_{S,0}:= \underset{\leftarrow}{\lim} \ \BB\otimes_\CC \OO_S/m^k
\subset {\HH}_f^{(0)}\hat{\otimes}_\CC\check{\OO}_{S,0}$ (actually, this is equal to ${\BB}{\otimes}_{\CC}
\check{\OO}_{S,0}$ due to the finite dimensionality of $\BB$) and
${\LL}\hat{\otimes}_{\CC} \check{\OO}_{S,0}:= \underset{\leftarrow}{\lim} \
\LL\otimes_\CC \OO_S/m^k\subset \HH_f\hat{\otimes}\check{\OO}_{S,0}$. Then, we
consider the images $e^{(f-F)/z}({\BB}\hat{\otimes}_{\CC} \check{\OO}_{S,0})$
and $e^{(f-F)/z}({\LL}\hat{\otimes}_{\CC} \check{\OO}_{S,0})$ in $\check{\HH}_F$
by the evolution morphism \eqref{eq:exponential}, and call them the formal
section and the formal opposite filtration on $\check{\HH}_F$ associated with
$\BB$ and $\LL$, respectively.

\smallskip Let us first justify the namings: the formal section and the formal
opposite filtration.

\begin{Proposition}
\label{FormalSection} \rm{(~\cite{LLS} Lemma 4.9)}
Let $\LL\in \MM_{o,f}$ be an opposite filtration.
Then,

1.  The formal evolution image $e^{(f-F)/z}({\LL}\hat{\otimes}_{\CC}
\check{\OO}_{S,0})$ of ${\LL}\hat{\otimes}_{\CC} \check{\OO}_{S,0}$ gives a
formal opposite filtration of $\check{\HH}_F$ in the sense that it defines a
splitting of $\check{\OO}_{S,0}$-module \be
\label{eq:Splitting} \check{\HH}_F\ = \ \check{\HH}_F^{(0)} \ \oplus \
e^{(f-F)/z}({\LL}\hat{\otimes}_{\CC} \check{\OO}_{S,0}) \ee satisfying the
formal opposite condition: \be
\label{eq:FormalOpposit} z^{-1} (e^{(f-F)/z}({\LL}\hat{\otimes}_{\CC}
\check{\OO}_{S,0}) ) \ \subset \ e^{(f-F)/z}({\LL}\hat{\otimes}_{\CC}
\check{\OO}_{S,0}).  \ee

2.  The intersection $\check{\HH}_F^{(0)} \cap
ze^{(f-F)/z}({\LL}\hat{\otimes}_{\CC} \check{\OO}_{S,0}) $ is a formal section
of $\check{\HH}_F^{(0)}$ in the sense that the natural projection to the
$z$-graded piece, \eqref{eq:Filter2} for $k=0$, induces an isomorphism:
\be
\label{eq:FormalSection} \check{\HH}_F^{(0)} \cap
ze^{(f-F)/z}({\LL}\hat{\otimes}_{\CC} \check{\OO}_{S,0}) \quad \simeq \quad
\check{\HH}_F^{(0)}/\check{\HH}_F^{(-1)} \quad (\simeq \check{\Omega}_F) .  \ee
\end{Proposition}
\begin{proof} 1. Let a $\CC$-basis of $\BB:=\HH^{(0)}_f\cap z\LL$ be presented
by forms $\phi_j\in p_{0*}(\Omega_{X}^n)[[z]]$ ($j=1,\cdots,\mu$). Consider the
$\OO_S$-morphism $r_F^{(0)}: \BB_F:=\oplus_{j=1}^\mu \OO_S [\pi_X^{-1}(\phi_j) ]
\to \Omega_F$ given in \eqref{eq:Filter2}.  Its restriction $r_f^{(0)}$ to the
origin $0\in S$ (i.e.\ to tensor with $\OO_{S,0}/m$) is the isomorphism
$\BB\simeq \Omega_f$ (since $\BB$ is a section). Then by Nakayama's lemma, the morphism  is isomorphic at the stalks at the origin.  Then the $\OO_S$-isomorphism between coherent sheaves at the  origin extends to an $\OO_S$-isomorphism in a neighborhood of the origin of
$S$. This implies that $\BB_F = \oplus_{j=1}^\mu \OO_S [\pi_X^{-1}(\phi_j) ] \
(\simeq\BB\otimes_\CC\OO_S$ depending on $\phi_j$'s) on the neighborhood is a
section of $\HH_F^{(0)}$. So,
$z^{-1}\BB_F[z^{-1}]$ is an opposite filtration of $\HH_F$ in a neighborhood of
$0\in S$.  In particular, we get a splitting
\begin{equation}
\label{eq:OL} \HH_F\otimes_{\OO_S}\OO_S/m^{k+1}
=\Big(\HH_F^{(0)}\otimes_{\OO_S}\OO_S/m^{k+1}\Big) \oplus
\Big(z^{-1}\BB_F[z^{-1}]\otimes_{ O_S} \OO_S/m^{k+1} \Big) \quad k\in\ZZ_{>0}
\end{equation} On the other hand, recall the isomorphism: $e^{(f-F)/z} :
\BB((z))\otimes_\CC\OO_S/m^{k+1} \to \HH_F\otimes_{\OO_S} \OO_S/m^{k+1}$
(Proposition \ref{FlatExtension1}). Let us show that the composition of this
isomorphism with the embedding $\LL\otimes_\CC\OO_S/m^{k+1}
(=z^{-1}\BB[z^{-1}]\otimes_\CC\OO_S/m^{k+1}) \subset
\BB((z))\otimes_\CC\OO_S/m^{k+1}$ and with the projection to the second factor
$z^{-1}\BB_F[z^{-1}]\otimes_{ O_S} \OO_S/m^{k+1}$ of \eqref{eq:OL} (i.e.\
quotient by the submodule $\HH_F^{(0)}\otimes_{\OO_S}\OO_S/m^{k+1}$, see the
following Remark 4.3) is an $\OO_S/m^{k+1}[z^{-1}]$-isomorphism
 \be
 \label{eq:polarPart}
 e^{(f-F)/z}\LL\otimes_\CC\OO_S/m^{k+1} \ \simeq \ z^{-1}\BB_F[z^{-1}]\otimes_{
O_S} \OO_S/m^{k+1}.
\ee
({\it Proof.} We note the facts

i) The map is  an $\OO_S/m^{k+1}[z^{-1}]$-homomorphism. 

ii) Both hand sides of the map are $\OO_S/m^{k+1}[z^{-1}]$-free modules.

iii) The map induces an isomorphism if one consider the map modulo the maximal
ideal $m$.

\noindent Then it is clear that the map itself is an
$\OO_S/m^{k+1}[z^{-1}]$-isomorphism.  \ $\Box$\ )

 Thus, the image $e^{(f-F)/z}(\LL\otimes_\CC\OO_S/m^{k+1})$ in
$\HH_F\otimes_{\OO_S} \OO_S/m^{k+1}$ defines a splitting 
$$
\HH_F\otimes_{\OO_S}
\OO_S/m^{k+1} = (\HH_F^{(0)}\otimes_{\OO_S}\OO_S/m^{k+1}) \oplus
e^{(f-F)/z}(\LL\otimes_\CC\OO_S/m^{k+1}). 
$$
Clearly these splittings for $k\in
\ZZ_{\ge0}$ form a projective system, and hence taking the limit for
$k\to\infty$, we obtain
the splitting 
\eqref{eq:Splitting}. The opposite condition (i.e.\ $z^{-1}$-stability) on
$e^{(f-F)/z}(\LL\otimes_\CC\OO_S/m^{k+1})$ obviously follows from that on $\LL$.

\smallskip 2. The formula \eqref{eq:FormalSection} follow from the same logic in
the first half of Proposition \ref{SectOppo} .
\end{proof}

\begin{remark} This quotient process in the proof of \eqref{eq:polarPart} is essentially necessary, since the evolution map
$e^{(f-F)/z}$ may not bring elements of $\LL\otimes_\CC\OO_S/m^{k+1} $ into
$z^{-1}\BB_F[z^{-1}]\otimes_{ O_S} \OO_S/m^{k+1}$. This is caused, since, in
order to express elements $(F_0-F)^l \phi_j$'s in a linear combination of
elements in $\phi_j\otimes \OO_S/m^{k+1}$, one need to calculate modulo the
image of $D_F$, which may create some terms having coefficients of positive
powers in $z$.
\end{remark}
 
 \begin{remark} The evolution image $e^{(f-F)/z}({\BB}\hat{\otimes}_{\CC}\check{\OO}_{S,0})$ of the section $\BB$ for $\LL$ does not give the section $\check{\HH}_F^{(0)} \cap ze^{(f-F)/z}({\LL}\hat{\otimes}_{\CC} \check{\OO}_{S,0})$ for the opposite filtration $e^{(f-F)/z}({\LL}\hat{\otimes}_{\CC} \check{\OO}_{S,0})$, since the image as a submodule of the direct sum  \eqref{eq:Splitting} may not be contained in the first factor, but just trancating its second factor part, it does give the section $\check{\HH}_F^{(0)} \cap ze^{(f-F)/z}({\LL}\hat{\otimes}_{\CC} \check{\OO}_{S,0})$. Actually, on the contrary, the evolution image of $\BB$ contains highly interesting information in the second factor.
  \end{remark}

\begin{remark} The naive relation $z^{-1}\BB[z^{-1}]=\LL$ inside $\HH_f$ does
not extends to the formal level, but we only have
$z^{-1}\big(\check{\HH}_F^{(0)} \cap ze^{(f-F)/z}({\LL}\hat{\otimes}_{\CC}
\check{\OO}_{S,0})\big)[z^{-1}] \underset{\not=}{\subset}
e^{(f-F)/z}({\LL}\hat{\otimes}_{\CC} \check{\OO}_{S,0})$ inside $\check{\HH}_F$
due to the phenomenon in Remark 4.1. However, the opposite filtration
$e^{(f-F)/z}({\LL}\hat{\otimes}_{\CC} \check{\OO}_{S,0}) $ can be recovered from
the section $\check{\HH}_F^{(0)} \cap ze^{(f-F)/z}({\LL}\hat{\otimes}_{\CC}
\check{\OO}_{S,0})$ by a ``formal completion" (see \S6 proof of Theorem 6.1).
\end{remark}

\medskip
The flatness of the formal evolution morphism $e^{(f-F)/z}$ (Propostion 4.1, 2.)
implies the following important stability property of the formal section and the
formal opposite filtration.

\begin{Proposition}
\label{t-stability} {\rm (\cite{LLS} Lemma 4.11.)}  The Gauss-Manin connection
$\nabla$ on $\check{\HH}_F$ over $\TT_S$ preserves 
the formal opposite filtration given in Proposition \ref{FormalSection}. 
That is, 
$$
\begin{array}{rcl}
\nabla_{\TT_S} \big( e^{(f-F)/z}({\LL}\hat{\otimes}_{\CC} \check{\OO}_{S,0}) \big)  \subset   e^{(f-F)/z}({\LL}\hat{\otimes}_{\CC} \check{\OO}_{S,0}). 
\end{array}
$$
We  call this property $\underline{t}$-stability of  the formal opposite filtration. 
\end{Proposition}
\begin{proof} Let an element $s$ of $e^{(f-F)/z}({\BB}\hat{\otimes}_{\CC}
\check{\OO}_{S,0})$ or $e^{(f-F)/z}({\LL}\hat{\otimes}_{\CC} \check{\OO}_{S,0})$
be given by a sequence $s_k=\sum_i e^{(f-F)/z} a_i\otimes b_i \bmod m^{k+1}$ for
$a_i\in \BB$ or $\in \LL$ and $b_i\in \OO_{S,0}/m^{k+1}$ for $k\in \ZZ_{\ge
0}$. Since $e^{(f-F)/z} a_i$ is flat \eqref{eq:flatness}, we have $\nabla_v s
\equiv \sum_i e^{(f-F)/z} a_i\otimes \partial_{\tilde v} b_i \bmod m^{k}$ where
the RHS converges to an element in $e^{(f-F)/z}({\BB}\hat{\otimes}_{\CC}
\check{\OO}_{S,0})$ or in $e^{(f-F)/z}({\LL}\hat{\otimes}_{\CC}
\check{\OO}_{S,0})$, respectively.
\end{proof}

The $z\partial_z$-stability property and the metric property of sections $\BB$
and opposite filtrations $\LL$ in $\HH_f$ are inherited to formal section and
formal opposite filtrations in $\check{\HH}_F$ as follows.

\begin{Proposition}
\label{FlatExtension} Let $\LL\in \MM_{o,f}$ be an opposite filtration
associated with a section $\BB\in \MM_{s,f}$.

1.  If $\LL$ is $z\partial_z$-stable (recall Proposition \ref{z-stable}), then
the formal opposite filtration $e^{(f-F)/z}({\LL}\hat{\otimes}_{\CC}
\check{\OO}_{S,0})$ is formally $z\partial_z$-stable in the sense: \be
\label{eq:z-stable} \nabla_{z\partial_z}
\big(e^{(f-F)/z}({\LL}\hat{\otimes}_{\CC} \check{\OO}_{S,0})\big) \ \subset \
e^{(f-F)/z}({\LL}\hat{\otimes}_{\CC} \check{\OO}_{S,0}).  \ee

2. If $\LL$ satisfies the metric condition (see Proposition \ref{Metric}), then
the formal opposite filtration $e^{(f-F)/z}({\LL}\hat{\otimes}_{\CC}
\check{\OO}_{S,0})$ satisfies the formal metric condition: \be
\label{eq:Metric2} \check{K}_F\big(e^{(f-F)/z}({\LL}\hat{\otimes}_{\CC}
\check{\OO}_{S,0}), \ e^{(f-F)/z}({\LL}\hat{\otimes}_{\CC}
\check{\OO}_{S,0})\big) \ \subset \ z^{-2+n} \big(\OO_S[z^{-1}])^\vee \ee and
the formal section $\check{\HH}_F^{(0)} \cap
ze^{(f-F)/z}({\LL}\hat{\otimes}_{\CC} \check{\OO}_{S,0}) $ satisfies the formal
purity condition: \be
\label{eq:Metric3} \check{K}_F\big(\check{\HH}_F^{(0)} \cap
ze^{(f-F)/z}({\LL}\hat{\otimes}_{\CC} \check{\OO}_{S,0}), \ \check{\HH}_F^{(0)}
\cap ze^{(f-F)/z}({\LL}\hat{\otimes}_{\CC} \check{\OO}_{S,0})\big ) \quad
\subset \quad z^n \check{\OO}_{S,0} .  \ee
\end{Proposition}
\begin{proof} Those follow formally from the definitions. 

1. Let us express an element of 
${\LL}\hat{\otimes}_{\CC} \check{\OO}_{S,0}$by the limit of a sequence of
elements of the form $\sum_i a_i^k\otimes b_i^k \bmod m^{k+1}$ for $a_i^k\in
\LL$ and $b_i^k\in \OO_S$ for $k=1,2,3,\cdots$. Then, using the first line of
\eqref{eq:Gauss-Manin},
$$
\begin{array}{ll} &\nabla_{z\partial_z}\big(e^{(f-F)/z}\sum_i a_i\otimes b_i
\bmod m^{k+1}\big)\! \\ = & e^{(f-F)/z} \big(\sum_i -((f-F)/z) a_i\otimes b_i
+(z\partial_z a_i)\otimes b_i -z^{-1}F a_i\otimes b_i \big) \bmod m^{k+1}\\ =
&e^{(f-F)/z} \big(\sum_i -(f/z) a_i +(z\partial_z a_i)\big) \otimes b_i \bmod
m^{k+1}\\ = & e^{(f-F)/z} \big(\nabla_{z\partial_z} a_i \big) \otimes b_i \bmod
m^{k+1}\\
\end{array}
$$
By assumption, $\nabla_{z\partial_z} a_i \in \LL$ so that the last line belongs
to $e^{(f-F)/z}({\LL}\hat{\otimes}_{\CC} \check{\OO}_{S,0})$, implying
\eqref{eq:z-stable}.

\smallskip 2.\ Since $K_F$ is $\OO_S$-bilinear \eqref{eq:HRP}, it induces a
$\OO_S/m^{k+1}$-bilnear form such that
$$
K_F\big(e^{(f-F)/z}({\LL}\hat{\otimes}_{\CC} \check{\OO}_{S,0}) \bmod m^{k+1}, \
e^{(f-F)/z}({\LL}\hat{\otimes}_{\CC} \check{\OO}_{S,0}) \bmod m^{k+1}\big)
\subset z^{-2+n} \OO_{S,0}[z^{-1}] \bmod m^{k+1}.
$$
This implies clearly \eqref{eq:Metric2}. Hence, combining with \eqref{eq:HRdeg},
we get \eqref{eq:Metric3}.
\end{proof}
\begin{remark} Suppose the opposite filtration $\LL$ in $\HH_f$ satisfy the
metric condition. Then, Proposition \ref{FlatExtension}, 2.\ implies
$\omega\big(e^{(f-F)/z}({\LL}\hat{\otimes}_{\CC}
\check{\OO}_{S,0}),e^{(f-F)/z}({\LL}\hat{\otimes}_{\CC}
\check{\OO}_{S,0})\big)=0$. Hence the decomposition \eqref{eq:Splitting} is
holonomic decomposition of $\check{\HH}_F$, since
$\omega(\check{\HH}_F^{(0)},\check{\HH}_F^{(0)})=0$ is automatic due to the
formal version of \eqref{eq:HRdeg} (see Remark \ref{simplectic}). But we shall
not use this fact explicitely.
\end{remark}

\begin{remark} As in the latter half of Remark 4.1, one should note again that
the range $z^{-2+n}\big(\OO_S[z^{-1}])^\vee$ (see \eqref{eq:Metric2})) of the
residue pairings $\check{K}_F$ on $e^{(f-F)/z}({\LL}\hat{\otimes}_{\CC}
\check{\OO}_{S,0})$ is strictly larger than
$z^{-2+n}\check{\OO}_{S,0}[z^{-1}]$. That is, the order of the pole in $z$ of
the coefficient of $t^k$ may increase  infinitely as $k\to \infty$.
\end{remark}

\section{Primitive forms without metric structure} In this section, we
introduce the notion of (formal or non-formal) primitive forms and then (formal
or non-formal) primitive forms without metric structure.

Let us recall the definition of primitive forms \cite{S1} in slightly reformed
way.  
\footnote {
Compared with \cite{S1}, in the present formulation as in
\cite{FromPrimitiveToFrobenius} and \cite{LLS}, we included the descendent parameter $z$ explicitly
so that the base space $\CC\times S$ of the family of vanishing cycles is one
dimensional higher than the base space $S$ of the unfolding in \cite{S1}. This
makes essentially two changes of the formulation: 1.\ $F_1$ action is changed to
$z\partial_z$-action in the second condition of (P3) and in (P4). 2.\ To kill
the redundancy of the base space, we need to add one more relation (P1) (which
was missing in \cite{FromPrimitiveToFrobenius} and \cite{LLS}, see \cite{Tcorrection}), formulated
as the triviality of primitive forms in primitive direction (actually, this
formula (P1) is the reason of calling $\partial_1$ primitive derivation or
primitive direction).  }

\begin{definition} {\bf Primitive Form. } An element $\zeta\in
\Gamma(S,\HH_F^{(0)})$ is called a \textit{primitive form} if it satisfies the following
conditions {\rm (P0)-(P4)}.

{\rm (P0) (Primitivity)}: The covariant derivations of $\zeta$ induces an
$\OO_S$-isomorphism:
$$
z\nabla\zeta \ : \ \TT_S \ \rightarrow \ \HH_F^{(0)}/\HH_F^{(-1)}\simeq
\Omega_F, \quad v\mapsto z\nabla_v\zeta \ \bmod \HH_F^{(-1)}.
$$

{\rm (P1) (Triviality in primitive direction)}: The $\zeta$ is trivial in the
primitive direction $\partial_1$. That is,
$$
z \nabla_{\partial_1}\ \zeta\ =\ \zeta .
$$

{\rm (P2) (Purity of higher residue pairings)}: The covariant derivations of $\zeta$ has the
following purity w.r.t. the residues.
$$  
K_F(\nabla_{\TT_S}\zeta,\nabla_{\TT_S}\zeta)\subset z^{n-1}\OO_S \ .
$$

{\rm (P3) (Holonomicity)}: The second covariant derivatives of $\zeta$ satisfies
the followings.
$$
\begin{array}{lccl} K_F(\nabla_{\TT_S}\nabla_{\TT_S}\zeta,\nabla_{\TT_S}\zeta) \
\in \ z^{n-2}\OO_S+ z^{n-1}\OO_S \ \ \\
K_F(\nabla_{z\partial_z}\nabla_{\TT_S}\zeta,\nabla_{\TT_S}\zeta) \ \in \
z^{n-2}\OO_S + z^{n-1}\OO_S .
\end{array}
$$

{\rm (P4) (Homogenity)}: There is a constant $r\in \CC$ such that
$$
\nabla_{z\partial_z+E}\ \zeta = \ r\zeta.
$$
\qquad\quad where $E$ is the Euler vector field \eqref{eq:UnitEuler}.
\end{definition}

\noindent
\begin{remark}\label{rem:Formal}  The definition is analytic in $S$ direction. However, it is still formal in $z$-variable.
\end{remark}
\begin{remark} \label{rem:Section} The condition {\rm (P0)} is equivalent to that the image of covariant
differentiations of $\zeta$
$$
\label{eq:Bzeta} \BB_{\zeta}\ := \ Im\Big( z\nabla \zeta: \TT_S \to
\HH^{(0)}_F\Big) \ \ = \ z\nabla_{ \TT_S}\zeta
$$
form a section of $\HH^{(0)}_F$.  Using this notation, the condition {\rm (P2)}
is equivalent to the following
$$  
\text{\rm (P2)$^*$ (Orthogonality)} \qquad \qquad\qquad
K_F(\BB_{\zeta},\BB_{\zeta})\subset z^{n+1}\OO_S \ . \qquad \qquad \qquad\qquad
\qquad\qquad\qquad
$$

\noindent Then, under the assumptions {\rm (P0)} and {\rm (P2)} ($\simeq$ {\rm
(P2)}$^*$), the condition {\rm (P3)} is equivalent to

\smallskip {\rm (P3)$^*$ (Holonomicity)$^*$}: The second covariant derivatives
of $\zeta$ satisfies the followings.
$$
\begin{array}{rcl}
  \nabla_{\TT_S}\nabla_{\TT_S}\zeta & \subset & z^{-2}\BB_{\zeta}
                                                + z^{-1}\BB_{\zeta} \\ 
\nabla_{z\partial_z}\nabla_{\TT_S}\zeta & \subset & z^{-2}\BB_{\zeta}+ z^{-1}\BB_{\zeta}.
\end{array}
$$

\end{remark}

\medskip
We now, removing the metric condition (P2) from the conditions (P0)-(P4) for a
primitive form,  introduce a concept of \textit{primitive forms without metric
condition} (see Remark 5.4 for this terminology), where we replace the condition
(P3) by (P3)$^*$ (see Remark 5.1 for this replacement).

\begin{definition} {\bf Primitive Form without metric structure. }
\label{PFwithoutmetric} 
An element $\zeta\in \Gamma(S,\HH_F^{(0)})$ is called a
\textit{primitive form without metric structure} if it satisfies the conditions {\rm
(P0)}, {\rm (P1)}, {\rm (P3)}$^{*}$ and {\rm (P4)} but may not necessarily
satisfy {\rm (P2)}$^{*}$.



\end{definition}

Next, let us recall the definition of formal primitive forms \cite{LLS}, Definition
4.6.

\begin{definition} 
\label{FPF}
{\bf Formal Primitive Form. } An element $\zeta\in
\check{\HH}_F^{(0)}$ is called a \textit{formal primitive form}, if it satisfies {\rm
(P0)$^\vee$}, {\rm (P1)}, {\rm (P2)$^\vee$}, {\rm (P3)$^\vee$} and {\rm (P4)},
where

\medskip {\rm (P0)$^\vee$}: Replace $\TT_S$ and $\Omega_F$ in {\rm (P0)} by
$\check{\TT}_S=\TT_S\otimes_{\OO_S}\check{\OO}_{S,0}$ and by
$\check{\Omega}_F=\Omega_F\otimes_{\OO_S}\check{\OO}_{S,0}$, respectively.

{\rm (P2)$^\vee$}: Replace $\TT_S$, $K_F$ and $\OO_S$ in {\rm (P2)} by
$\check{\TT}_S$, $\check{K}_F$ and by $\check{\OO}_{S,0}$, respectively.

{\rm (P3)$^\vee$}: Replace $\TT_S$, $K_F$ and $\OO_S$ in {\rm (P3)} by
$\check{\TT}_S$, $\check{K}_F$ and by $\check{\OO}_{S,0}$, respectively.

\end{definition}

\medskip

Finally, let us introduce a concept of formal primitive forms without metric
structure

\begin{definition}
\label{FPFwithoutmetric} 
{\bf Formal Primitive Form without metric structure. }
An element $\zeta\in \check{\HH}_F^{(0)}$ is called a \textit{formal primitive form
without metric structure}, if it satisfies {\rm (P0)$^\vee$}, {\rm (P1)}, {\rm
(P3)$^{* \vee}$} and {\rm (P4)} but may not necessarily satisfy {\rm
(P2)$^{\vee}$}, where

\smallskip {\rm (P3)$^{*\vee}$}: Replace $\TT_S$ and $\BB_{\zeta}$ in {\rm
(P3)$^*$} by $\check{\TT}_S$ and by $\check{\BB}_{\zeta}$, respectively, where
we set the formal section by
$$
\check{\BB}_{\zeta}\ := \ Im(z\nabla \zeta : \check{\TT}_S\to
\check{\HH}^{(0)}_F) \ \ = \ z\nabla_{ \check{\TT}_S}\zeta.
$$
\end{definition}

\medskip
\begin{remark} An obvious remark is that the embedding \eqref{eq:inclusion} send
a primitive form $\zeta$ without metric structure in $\Gamma(S,
\HH_F^{(0)})$ to a formal primitive form $\zeta$ without metric
structure in $\check{\HH}_F^{(0)}$.
Is the converse, i.e.\ all formal primitive
forms without metric structure are induced from some analytic primitive
form without metric structure, true?
\end{remark}
\begin{remark} Note that the condition {\rm (P3)} alone cannot be equivalent to
{\rm (P3)$^*$} or {\rm (P3)$^{**}$}, but the combination {\rm (P2)} and {\rm
(P3)} is equivalent to the combination {\rm (P2)} and {\rm (P3)$^*$} (see Proof
4. of {\rm Theorem \ref{Theorem}}).
\end{remark}
\begin{remark} In Definitions \ref{PFwithoutmetric} and \ref{FPFwithoutmetric},
the terminology ``without metric structure" is used in a weak sense ``with or
without metric structure". That is, a primitive form $\zeta$ without metric
structure may eventually satisfy the metric condition {\rm (P2)} so that it can
actually be a primitive form, unless it is explicitly stated that $\zeta$ does
not satisfy {\rm (P2)}.
\end{remark}

\section{Construction of formal primitive forms with
or \label{sec:Construction}\\ \qquad\quad {without metric structure}}

In order to formulate the first main result, Theorem 6.1, of the present paper, we
prepare a notion of a {\it good pair} (\cite{LLS} Def. 4.15, see Remark 6.1).

\begin{definition}
\label{GoodPair} A pair $(\LL, \zeta_0)$ of a $z\partial_z$-stable opposite
filtration $\LL$ of $\HH_f$ and an element $\zeta_0 \in \BB:=\HH_f^{(0)}\cap
z\LL$ is called a \textit{good pair} (without metric condition) if it satisfies the
following (1) and (2).

(1) (primitivity) The image of $\zeta_0$ in $\Omega_f$ by $r_f^{(0)}$ (see
\eqref{eq:Filter2}) generates $\Omega_f$ over $\OO_{C_f}$.

(2) (homogeneity) There is a constant $r\in\CC$ such that
$\nabla_{z\partial_z}\zeta_0-r\zeta_0 \in \LL$.

\noindent If $\LL$ satisfies further the metric condition (c.f.\ {\rm
Proposition \ref{Metric}}), the pair $(\LL,\zeta_0)$ is called a good pair with
metric condition (see Remark 3.1).
\end{definition}

\begin{remark} According to the correction of Definition of primitive forms (see
2. of the Footnote 3  in \S5), we correct the definition of a good pair from
\cite{LLS}. Namely, we took $\zeta_0$ from $\HH_f^{(0)}$ in \cite{LLS}, but, in the present paper, 
we assume more strongly that $\zeta_0$ belongs to $\BB:=\HH_f^{(0)}\cap z\LL$.
\end{remark}

Theorem 4.16 in~\cite{LLS} gives a construction of formal primitive forms from a good
pair with metric condition.  We confirm in the following Theorem that the same construction
works for the case without metric condition.

\begin{theorem}
\label{Theorem} Let $F$ be a universal unfolding of a function $f$ with an
isolated critical point at the origin (See Definitions \ref{Unfolding} and
\ref{UniversalUnfolding}). Then, we have bijections 
\ba
\label{eq:PFwithoutmetric} \{\text{good pairs }(\LL,\zeta_0) \text{ on }
\HH_f\} \ \simeq \ \{\text {formal primitive forms in }
\check{\HH}_F^{(0)} \text{ without metric structure}\} \\ 
\ea 
\ba
\label{eq:PFwithmetric} \{\text{good pairs }(\LL,\zeta_0) \text{ on }\HH_f
\text{ with metric structure}\} \ \simeq \ \{\text{formal primitive forms in
}\check{\HH}_F^{(0)}\} 
\ea 
The correspondence $ (\LL,\zeta_0) \ \mapsto
\zeta_+$ from LHS to RHS is given as follows: apply the evolution morphism to
$\zeta_0$ so that we obtain an element $e^{(f-F)/z} \zeta_0$ in $
\check{\HH}_F$.  Then, apply the splitting \eqref{eq:Splitting} to this element
so that we have the decomposition: 
\be
\label{eq:Decomposition} e^{(f-F)/z} \zeta_0 \quad = \quad \zeta_+ \ + \ \zeta_-
\ee 
for $\zeta_+ \in \check{\HH}_F^{(0)}$ and $\zeta_- \in
e^{(f-F)/z}({\LL}\hat{\otimes}_{\CC} \check{\OO}_{S,0})
$.  
\end{theorem}
\begin{proof} In the following, the proof is devided in Steps 1-13.  Compared with the proof for Theorem 4.16 in~\cite{LLS}, since we added a new
axiom, say (P1) ($\partial_1$-triviality), on the primitive form side, and an
additional condition $\zeta_0\in z\LL$ on the primitive element side, we need some additional arguments, in particular, the Steps 2, and 13. We
added also Proposition \ref{zetaminus} to complete the Step 13 of the proof.  In
order to clarify the role of metric conditions, we take caution on the treatment
of the conditions (P3) and (P3)$^*$ and distinguish the places where and when
the metric condition is necessary or not.

Before coming to the proof of Theorem \ref{Theorem}, we prepare below a
proposition on $\zeta_-\in e^{(f-F)/z}({\LL}\hat{\otimes}_{\CC}
\check{\OO}_{S,0})$ for a use in the proof. Actually, beyond this use, $\zeta_-$ itself is of interest and we shall come back to $\zeta_-$ in \S10.

\begin{Proposition}
\label{zetaminus} 
Let $(\LL,\zeta_0)$ be a good pair, and consider the decomposition \eqref{eq:Decomposition}.  Set 
\be
\label{eq:zetaminus} 
\zeta_-\quad = \quad e^{(f-F)/z}\big( \sum_{j=1}^\infty
z^{-j} u_j \big) 
\ee 
for $\sum_{j=1}^\infty z^{-j} u_j\in \LL\hat{\otimes}_\CC
\check{\OO}_{S,0}$ with $u_j\in \BB\otimes_\CC \check{\OO}_{S,0}$ $(j\in
\ZZ_{>0})$. 
\footnote{
For the same reason as in Remark 4.1, the elements in
$\LL\hat{\otimes}_\CC \check{\OO}_{S,0}$ may not always be expanded in  Laurent
series in $z$ of bounded order of poles.  Actually, we shall see in \S10 that the coefficients in $u_j$ are non-trivial polynomials of degree $j$ with respect to the flat coordinates introduced in \S8.
 } 
Then, 
\be
 \label{eq:uj} 
 u_j \ \in \ \BB\otimes_\CC m^j\check{\OO}_{S,0}.  
 \ee
 \end{Proposition}
\begin{proof} 
Applying the inversion $e^{(F-f)/z}$ of the evolution map to the
formula \eqref{eq:Decomposition}, we obtain 
\be
\label{eq:Jfunction} 
\quad \zeta_0 \quad = \quad e^{(F-f)/z}\zeta_+ +
\sum_{j=1}^\infty z^{-j} u_j 
\ee 
In $\HH_f\hat{\otimes}_{\CC} \check{\OO}_{S,0}$, we rewrite: $ \sum_{j=1}^\infty z^{-j} u_j = \zeta_0 \ - \ e^{(F-f)/z}\zeta_+$.  
We want to represent the projection image of the RHS of this element in  $\HH_f\hat{\otimes}_{\CC} \check{\OO}_{S,0}/m^{k+1}$ for $k\ge0$ by an element of 
 $\oplus_{l=1}^{k} z^{-l} \hat{\BB}\otimes_\CC m^l \check{\OO}_{S,0}$ 
modulo the image of
$D_f \big(\Gamma(X,\Omega_X^{n})((z))\hat \otimes_{\CC} \check{\OO}_{S,0}\big)$, where we
recall that $\hat\BB\subset \Gamma(X,\Omega_X^{n+1})[[z]]$ is a $\mu$-dimensional
$\CC$-subspace representing $\BB$ 
 (see paragraphs following \S3 Definition 3.1) 
and note that the action of
$D_f$ commutes with the coefficients in $\check{\OO}_{S,0}$.  

Let $\hat
\zeta_0\in \Gamma(X,\Omega_X^n)[[z]]$ and $\hat{\zeta}_+\in
\Gamma(Z,\Omega_{Z/S}^n)[[z]]$ be representatives of $\zeta_0$ and $\zeta_+$,
respectively, where we may assume that $\hat \zeta_0=\hat \zeta_+\mid_X$ (see
Proposition \ref{FlatExtension1} and its Corollary). Therefore, we have an
expression $\hat{\zeta}_+=\pi_X^{-1}(\hat{\zeta}_0)(1+a(\underline x,\underline
t,z))$ where $a=\sum_{|I|>0} a_I(\underline{x},z)\underline{t}^I$ is a 
power series in $\underline{t}$ with respect to the multi-index
$I=(i_1,\cdots,i_\mu)\in (\ZZ_{\ge0})^\mu$ ($|I|=\sum_{j=1}^\mu i_j$) 
and $a_I(\underline{x},z)\in \Gamma(X,\OO_X)[[z]]$ (by shrinking $X$ and $Z$ if necessary). 
Owing to the expression
\eqref{eq:Infinitesimal}, we lift $ \zeta_0 \ - \ e^{(F-f)/z}\zeta_+ \ \bmod  m^{k+1}$ to an
element in $\Gamma(X,\Omega_X^n)[[z]] \hat\otimes_\CC \check{\OO}_{S,0} \bmod m^{k+1}$:
 $$
\hat{\zeta}_0 \ - \ 
\sum_{l=0}^\infty \frac{(F(\underline x,\underline t)-F_0(\underline x))^l}{l!z^l}
\hat{\zeta}_0 (1+a(\underline x,\underline t,z))\qquad \bmod m^{k+1}\\
 $$
 Since $F-F_0$
vanishes along $X$, i.e.\ it is contained in $m\Gamma(Z,\OO_Z)$, the terms of summation index $l$ larger than $k$ vanish $\bmod\ m^{k+1}$. Hence, the expression is a finite sum till $l=k$.  
Let us develop this expression in the form $\sum_{K\in
(\ZZ_{\ge0})^\mu,0<|K|\le k} \alpha_{K}(\underline{x},z) \underline{t}^K$.  
The coefficient of this lifting is given by 
$$
 \alpha_{K}(\underline{x},z) = -\sum_{\substack{I,J\in (\ZZ_{\ge0})^\mu, \\ I+J=K}} \big( \sum_{l=0}^{k-|I|} (\text{the coefficient
of $\underline{t}^J$ in $\frac{(F(\underline x,\underline t)-F_0(\underline x))^l}{l!z^l}\hat{\zeta}_0 $}) \big)  \times a_I(\underline{x},z) 
$$
where we set $a_0:=0$. 
In view that  $a_I$ and $\hat{\zeta}_0$ have no-poles in $z$, we see that the order of pole in $z$ in the term for a fixed $J$ and $I=K-J$ is bounded by $l\le |J|\le |K|$. That is,
$\alpha_{K}(\underline{x},z) \in z^{-|K|}\Gamma(X,\Omega_X^{n+1})[[z]]$.
 We apply the decomposition \eqref{eq:V((z))} to $\alpha_{K}$
 so that it belongs to the direct sum
 $ z^{-|K|}\hat{\BB}[[z]] \
\oplus \ D_f\big(z^{-|K|}\Gamma(X,\Omega^{n}_X)[[z]]\big)$. Separating the principal part of the  Laurent series part in $z^{-|K|}\hat{\BB}[[z]]=\oplus_{m=1}^{|K|}z^{-m}\BB\oplus \BB[[z]]$ of $\al_K$, we decompose 
$ \alpha_{K} \equiv \sum_{m=1}^{|K|}  z^{-m}\beta_{K,m} +\beta_{K,0}  \bmod D_f\big(z^{-|K|}\Gamma(X,\Omega^{n}_X)[[z]]\big)$ for $\beta_{K,m}\in \hat \BB$ ($m=1,\cdots,|K|$) and $\beta_{K,0} \in \hat\BB[[z]]$. 
Multiply the monomial $\underline{t}^K$ to $\alpha_{K}(\underline{x},z)$ and sum up by the index $K$. Since it is a finite sum, we resum it according to the order $l$ of the pole in $z$,  and we obtain the expression 
$$
\sum_{0<|K|\le k} \!\!\! \al_K(\underline{x},z) \underline{t}^K= \sum_{l=1}^k z^{-l} \big( \!\! \sum_{l \le |K|\le k}  \!\! \beta_{K,l} \underline{t}^K \big) \ + \!\! \sum_{0 < |K|\le k} \!\! \beta_{K,0}  \underline{t}^K \ \   \bmod D_f\big(z^{-k}\Gamma(X,\Omega^{n}_X)[[z]]  \hat\otimes_{\CC} \check{\OO}_S/ m^{k+1}\big) 
$$
in the module  $ z^{-k}\Gamma(X,\Omega_X^{n+1})[[z]] \ \hat \otimes_{\CC} \ \check{\OO}_S/m^{k+1}$.  
Here, the $l$-th ($l=1,\cdots,k$)  component of the first term $\sum_{l \le |K|\le k} \beta_{K,l} \underline{t}^K$ belongs to $\sum_{l\le |K|\le k} \BB\otimes_\CC m^{|K|} \check{\OO}_S/ m^{k+1} =
\BB\otimes_\CC m^l \check{\OO}_S/ m^{k+1}$ and the second term 
$\sum_{0 < |K|\le k} \beta_{K,0} \underline{t}^K$ belongs to $\BB[[z]] \hat\otimes_{\CC} \check{\OO}_S/ m^{k+1}$.
On the other hand, we recall that the sum $\sum_{0<|K|\le k} \al_K(\underline{x},z) \underline{t}^K$ is a representative of the class $ \sum_{j=1}^\infty z^{-j} u_j  \bmod m ^{k+1}$. This implies
$$
u_j  \equiv \sum_{j \le |K|\le k} \beta_{K,j} \underline{t}^K 
\quad \text{and} \quad  \sum_{0 < |K|\le k} \beta_{K,0} \underline{t}^K\equiv 0 
\qquad \bmod \ \ m^{k+1}
$$
(actually, this is the constraint on $\hat{\zeta}_+$). 
In particular, we see that, for any infinitesimal neighborhood $\OO_{S,0}/m^{k+1}$
($k\in\ZZ_{\ge0}$), $u_j$ belongs to $\BB\otimes_\CC m^j\check{\OO}_{S,0}$.
\end{proof}

Let us return to the proof of Theorem \ref{Theorem}. The proof consists of the
following steps 0-13.
In the first steps 0 - 6., we show that the correspondence 
\be
\label{eq:Phi} \Phi \ :\ (\LL,\zeta_0)\quad \mapsto \quad \zeta_+ 
\ee 
defines
a map from the LHS to the RHS of \eqref{eq:PFwithoutmetric} and \eqref{eq:PFwithmetric}. In the steps 7 - 13., we  construct the inverse map $\Psi$ of $\Phi$.

Step 0.  {\it For any good pair $(\LL,\zeta_0)$, the associated $\zeta_+$ satisfies
the primitivity (P0)$^\vee$.}
\begin{proof} We need to show that the $\OO_{S,0}/m^{k+1}$-morphism
$z\nabla\zeta_+: \TT_S/ m^{k+1}\TT_S \to \Omega_F/ m^{k+1}\Omega_F$ between the
$\OO_{S,0}/m^{k+1}$-free modules is an isomorphism for any $k\in \ZZ_{\ge
0}$. Recall the description of the Gauss-Manin connection \eqref{eq:Gauss-Manin} so that
we have $z\nabla_v\zeta_+ \bmod \check{\HH}_F^{(-1)} \equiv [\partial_{\tilde
v}F \hat \zeta_+]$ where $\hat \zeta_+ \in
\big(p_*(\Omega^{n+1}_{Z/S}[[z]]\big)^{\vee}$ is a lifting of $\zeta_+$ to a relative
differential form. However, this morphism for $k=0$ is an isomorphism since 1)
$e^{(f-F)/z}\zeta_0 \bmod m \equiv \zeta_0$ implying that $\zeta_+ \bmod m
\equiv \zeta_0$, and 2) $\zeta_0$ satisfies the primitivity (1) in Definition
\ref{GoodPair}. Then, Nakayama Lemma implies the isomorphism over
$\OO_S/m^{k+1}$ for any $k\in \ZZ_{>0}$ and hence over $\check{\OO}_{S,0}$.
\end{proof}

Step 1. {\it For any good pair $(\LL,\zeta_0)$, the associated $\zeta_+$ gives the
following $\check{\OO}_{S,0}$-isomorphism \be
\label{eq:ZetaPlus} z\nabla\zeta_+\ : \quad \check{\TT}_S \quad \simeq \quad
\check{\HH}_F^{(0)} \cap z e^{(f-F)/z}({\LL}\hat{\otimes}_{\CC}
\check{\OO}_{S,0})
 \ee
 where we recall that the RHS is the formal section of $\check{\HH}_F^{(0)}$
given in Proposition \ref{FormalSection}, 2.}
\begin{proof} Let $v \in \check{\TT}_S$ ($i=1,2$) and let
$e^{(f-F)/z}\zeta_0=\zeta_{+}+\zeta_{-}$ be the splitting 
\eqref{eq:Splitting}.  Since $e^{(f-F)/z}\zeta_0$ is flat (Proposition
\ref{FlatExtension1}, 3.), we have $z\nabla_{v}\zeta_+=-z\nabla_{v}\zeta_-$,
where the LHS belongs to $\check{\HH}_F^{(0)}$ (transversality: 2. ii) of
Proposition \ref{GMHR}) and the RHS belongs to $z
e^{(f-F)/z}({\LL}\hat{\otimes}_{\CC} \check{\OO}_{S,0})$ (see Proposition
\ref{t-stability}). That is, \be
 \label{eq:ss} z\nabla_{v}\zeta_+\quad \in\quad \check{\HH}_F^{(0)} \cap z
e^{(f-F)/z}({\LL}\hat{\otimes}_{\CC} \check{\OO}_{S,0})
 \ee
 Combining isomorphisms (P0)$^\vee$ (above Step 0.) and
\eqref{eq:FormalSection}, we obtain the isomorphism \eqref{eq:ZetaPlus}.
 \end{proof}

Step 2.  {\it For any good pair $(\LL,\zeta_0)$, the associated $\zeta_+$ satisfies
the $\partial_1$-triviality} (P1).
\begin{proof} Since $\zeta_0\in \BB\subset z\LL$, one has $\zeta_++\zeta_-
=e^{(f-F)/z}\zeta_0 \in z e^{(f-F)/z}({\LL}\hat{\otimes}_{\CC}
\check{\OO}_{S,0})$. On the other hand, by definition, $\zeta_-\in
e^{(f-F)/z}({\LL}\hat{\otimes}_{\CC} \check{\OO}_{S,0})$. Then, in view of the
opposite condition \eqref{eq:FormalOpposit} of the formal opposite filtration,
we see that $\zeta_-$ and hence $\zeta_+=e^{(f-F)/z}\zeta_0-\zeta_-$ belongs to
$z e^{(f-F)/z}({\LL}\hat{\otimes}_{\CC} \check{\OO}_{S,0})$. Thus, $\zeta_+$
belongs to the RHS of \eqref{eq:ZetaPlus}, implying that there exists an unique
element $v\in \check{\TT}_S$ such that $z\nabla_v\zeta_+=\zeta_+$. This means,
by the morphism $r_F^{(0)}$ \eqref{eq:Filter2}, we obtain the relation
$(\tilde{v}F\mid_{C_F})\cdot \zeta_+= \zeta_+$ in $\check{\Omega}_F$. On the other hand,
since $\zeta_+\bmod m =\zeta_0$ is a generating element in $\Omega_f$
(Definition \ref{GoodPair} (1)), $\zeta_+$ is also a generating element in
$\check{\Omega}_F$. Then the relation $(\tilde{v}F\mid_{C_F}\! -1) \zeta_+=0$ implies
$\tilde{v}F\mid_{C_F}\! -1=0$. That is, by Definition \ref{UniversalUnfolding}, $v$ is
the primitive vector field $\partial_1$.
 \end{proof}
 
 Step 3. {\it Suppose that a good pair $(\LL,\zeta_0)$ satisfies the metric condition
(see Definition \ref{GoodPair} and Proposition \ref{Metric}), then the associated $\zeta_+$ satisfies the
orthogonality (P2)$^\vee$.}
 \begin{proof} This is an immediate consequence of above Step 1. combined with
the purity formula \eqref{eq:Metric3}.
 \end{proof}

Step 4.  {\it For any good pair $(\LL,\zeta_0)$, the associated $\zeta_+$ satisfies
the holonomicity (P3)$^{*\vee}$.}
\begin{proof} For $v_1,v_2\in \check{\TT}_S$, we have
$\nabla_{v_1}\nabla_{v_2}\zeta_+=-\nabla_{v_1}\nabla_{v_2}\zeta_-$. Applying the
transversality of $\nabla$ (2. i) of Proposition \ref{GMHR}) on the LHS and the
$\underline{t}$-stability of $e^{(f-F)/z}({\LL}\hat{\otimes}_{\CC}
\check{\OO}_{S,0})$ (Proposition \ref{t-stability}) on the RHS, we see that it is contained
in $z^{-2}\check{\HH}_F^{(0)}$ and $ e^{(f-F)/z}({\LL}\hat{\otimes}_{\CC}
\check{\OO}_{S,0})$, whose intersection decomposes
$$
 z^{-2}\check{\HH}_F^{(0)}\cap
e^{(f-F)/z}({\LL}\hat{\otimes}_{\CC} \check{\OO}_{S,0}) =
z^{-2}\big(\check{\HH}_F^{(0)}\cap ze^{(f-F)/z}({\LL}\hat{\otimes}_{\CC}
\check{\OO}_{S,0})\big) \oplus z^{-1}\big(\check{\HH}_F^{(0)}\cap
ze^{(f-F)/z}({\LL}\hat{\otimes}_{\CC} \check{\OO}_{S,0})\big)
$$
into two graded pieces. Then, applying Step 1.\ to each piece, this module is expressed as 
$
z^{-1}\nabla_{\check{\TT}_S}\zeta_+\oplus \nabla_{\check{\TT}_S}\zeta_+
$.
Then the first formula of (P3)$^{*\vee}$ follows.

Next, take $v\in \TT_S$. Then,
$\nabla_{z\partial_z}\nabla_{v}\zeta_+=-\nabla_{z\partial_z}\nabla_{v}\zeta_-$
belongs to $z^{-2}\check{\HH}_F^{(0)}\cap e^{(f-F)/z}({\LL}\hat{\otimes}_{\CC}
\check{\OO}_{S,0})$ (recall irregularity of $\nabla_{\partial_z}$: 2.\ iii) of
Proposition \ref{GMHR} and the $z\partial_z$-stability of
$e^{(f-F)/z}({\LL}\hat{\otimes}_{\CC} \check{\OO}_{S,0})$ Proposition
\ref{FlatExtension}). As in the first half, the intersection decomposes, and we obtain the second formula of (P3)$^{*\vee}$.
\end{proof}

Step 5. {\it Suppose that a good pair $(\LL,\zeta_0)$ satisfies the metric condition
(see Definition \ref{GoodPair}), then the associated $\zeta_+$ satisfies the
holonomicity (P3)$^\vee$.}
\begin{proof} The metric condition implies the purity (orthogonality) of higher
residue pairings \eqref{eq:Metric3}. Then, the conditions in (P3)$^{*\vee}$ are
reformulated into those in (P3)$^{\vee}$.
\end{proof}

Step 6. {\it For any good pair $(\LL,\zeta_0)$, the associated $\zeta_+$ satisfies
the homogeneity (P4).}
\begin{proof}
We calculate $\big(\nabla_{z\partial_z+E} -r\big)e^{(f-F)/z}\zeta_0$ by two different ways. 

1.  The flatness of the evolution map and the commutativity of
$\nabla_{z\partial_z}$ with the evolution map (3. and 5. of Proposition
\ref{FlatExtension1}) implies $\big(\nabla_{z\partial_z+E}-r\big)
e^{(f-F)/z}\zeta_0=e^{(f-F)/z}\big( (\nabla_{z\partial_z}-r)\zeta_0\big)$. Then, the homogeneity of $\zeta_0$ ((2) of Definition \ref{GoodPair}) implies that the last term is contained in 
$e^{(f-F)/z} \LL\subset e^{(f-F)/z}({\LL}\hat{\otimes}_{\CC} \check{\OO}_{S,0})$.
That is, the image is contained in the second splitting factor of \eqref{eq:Splitting}.

2. We remark that the splitting \eqref{eq:Splitting} is preserved by the action of $\nabla_{z\partial_z+E} -r$, since (1) the filter
$\check{\HH}_F^{(0)}$ is preserved due to 2. iv) of Proposition \ref{GMHR}, and
(2) $e^{(f-F)/z}({\LL}\hat{\otimes}_{\CC} \check{\OO}_{S,0})$ is preserved due
to the $z\partial_z$-stability (1. of Proposition \ref{FlatExtension}) and
$\underline{t}$-stability (Proposition \ref{t-stability}).  In particular, the action of $\nabla_{z\partial_z+E} -r$.  on the decomposition \eqref{eq:Decomposition} preserves each factor to each factor.

Comparing 1. and 2., we see that the action of $\nabla_{z\partial_z+E} -r$ on the first factor $\zeta_+$ should be zero. 
\end{proof}

\noindent {\it Note.} Acooding to the homogeneity of $\zeta_0$ (Definition
\ref{GoodPair}), set $G:=\nabla_{z\partial_z}\zeta_0-r\zeta_0 \in \LL$. Then,
the above argument in 6. shows that $\nabla_{E+z\partial_z}\zeta_-=r\zeta_- +
e^{(f-F)/z} G$.

\medskip Above 0.-6. altogether show that one direction $\Phi$ \eqref{eq:Phi} of
the correspondences \eqref{eq:PFwithoutmetric} and \eqref{eq:PFwithmetric} are
well-defined.  In the following 7-13., we construct the inverse correspondence
$\Psi: \zeta_+ \mapsto (\LL,\zeta_0)$ for a given formal primitive form $\zeta_+
\in \check{\HH}^{(0)}_F $ which may or may not be equipped with metric
structure (recall Definitions \ref{FPF} and \ref{FPFwithoutmetric}).

\medskip Step 7.  We first consider a set $\check{\BB}_{\zeta_+}$ of covariant
differentiations of $\zeta_+$.  
\be
\label{eq:Bzeta} \check{\BB}_{\zeta_+}\ := \ Im( z\nabla \zeta_+:
\check{\TT}_S\to \check{\HH}^{(0)}_F) \ = \ z\nabla_{ \check{\TT}_S}\zeta_+.
\ee 
By the primitivity (P0)$^\vee$, $\check{\BB}_{\zeta_+}$ is a section to the
projection $\check{\HH}_F^{(0)}\to \check{\Omega}_F$ \eqref{eq:Filter2} in the
sense that $\check{\BB}_{\zeta_+}$ is a $\check{\OO}_{S,0}$-submodule of
$\check{\HH}_F^{(0)}$ which is isomorphic to $\check{\Omega}_F$ by the
projection $\check{r}_F^{(0)}$.

 Therefore, for each $k\in \ZZ_{\ge0}$, we have identifications (c.f.\
Definition 3.1,1) \be
\label{eq:k-th} \HH_F /m^{k+1}\HH_F
=(\check{\BB}_{\zeta_+}/m^{k+1}\check{\BB}_{\zeta_+})((z)) \quad \text{and}
\quad \HH_F^{(0)} /m^{k+1}\HH_F^{(0)}
=(\check{\BB}_{\zeta_+}/m^{k+1}\check{\BB}_{\zeta_+})[[z]] .  \ee

Then, we consider the formal opposite filtration of $\check{\HH}_F$ (this
terminology shall be justified in the following Step 9.) by considering the projective
system of $\OO_S/m^{k+1}\OO_S$-submodule of $\HH_F /m^{k+1}\HH_F$ for
$k\in\ZZ_{\ge0}$ \be
\label{eq:Oppk} \LL_{\zeta_+}^k \ : = \
z^{-1}(\check{\BB}_{\zeta_+}/m^{k+1}\check{\BB}_{\zeta_+})[z^{-1}] \ee which is
a splitting factor in $\HH_F /m^{k+1}\HH_F =\HH_F^{(0)}
/m^{k+1}\HH_F^{(0)}\oplus \LL_{\zeta_+}^k $ for $k\in \ZZ_{\ge0}$.

Before going further to study the limit of this projective system, we formulate,
using the particular case $k=0$ of \eqref{eq:Oppk}, the statement of the proof
of Theorem \ref{Theorem}.

\medskip 
Step 8. {\it Let $\zeta_+$ be a formal primitive form without metric
structure. Consider the pair 
$$
\LL^0_{\zeta_+}:= \eqref{eq:Oppk}
\text{ for  $k=0$ \quad and \quad } 
\zeta_+^0:=\zeta_+\! \bmod\! m \ \in \HH_F^{(0)}/m\HH_F^{(0)}
\simeq \HH_f^{(0)}
$$
(recall {\it Note} after the proof of Theorem
\ref{FlatExtension1}). Then, the correspondence
$$
\Psi\ : \ \zeta_+ \quad \mapsto \quad (\LL^0_{\zeta_+},\zeta_+^0)
$$ 
gives a map from the RHS to the LHS of \eqref{eq:PFwithoutmetric} and \eqref{eq:PFwithmetric}, which is the inverse of $\Phi$.}
\begin{proof} The proof is divided into Steps 9-13.
Precisely, the fact that $\LL^0_{\zeta_+}$ is a $z\partial_z$-stable opposite
filtration of $\HH_f$ without metric condition is a particular case of
Step 9.  The facts that $\zeta_+^0$ belongs to $z\LL_{\zeta_+}^0$ and is primitive and
homogeneous are shown in Steps 10 and 11, respectively. The fact that $\Psi$ is
the inverse to $\Phi$ is shown in Step 13.
\end{proof}

\medskip Step 9. {\it The $\OO_S/m^{k+1}\OO_S$-modules \eqref{eq:Oppk} for
$k\in\ZZ_{\ge0}$ form a projective system so that the limit
$$
\check{\LL}_{\zeta_+}^\infty\quad := \quad \underset{\leftarrow}{\lim}\
\LL_{\zeta_+}^k
$$
is a formal opposite submodule in $\check{\HH}_F$ in the sense of
\eqref{eq:Splitting2} below such that $\check{\BB}_{\zeta_+}=
\check{\HH}_F^{(0)}\cap z\check{\LL}_{\zeta_+}^\infty$. The
$\check{\LL}_{\zeta_+}^\infty$ is $z\partial_z$-stable and
$\underline{t}$-stable.  If, further, $\zeta_+$ is a formal primitive form (that
is, it further satisfies the metric condition (P2)$^\vee$), then the
$\check{\LL}_{\zeta_+}^\infty$ satisfies the formal metric condition (see
Proposition \ref{Metric}).}

\begin{proof} The fact that $\LL_{\zeta_+}^k$ is a splitting factor
$\OO_{S,0}/m^{k+1}[z^{-1}]$-submodule of $\HH_F /m^{k+1}\HH_F$ forming a
projective system is obvious from the descriptions \eqref{eq:k-th}. These imply
that $\LL_{\zeta_+}^\infty$ is a opposite submodule in $\check{\HH}_F$ in the
sense that \be
\label{eq:Splitting2} 
\check{\HH}_F=\check{\HH}^{(0)}_F\oplus
\check{\LL}_{\zeta_+}^\infty \quad \text{ and} \quad
z^{-1}\check{\LL}_{\zeta_+}^\infty \subset \check{\LL}_{\zeta_+}^\infty.  \ee
The $z\partial_z$-stabilty (see Proposition \ref{z-stable}) and the
$\underline{t}$-stability (see Proposition \ref{t-stability}) of $
\check{\LL}_{\zeta_+}^\infty$ follow from those of $ \LL_{\zeta_+}^k$, which
follow from the condition (P3)$^*$
 on the primitive form $\zeta_+$.

Suppose further that $\zeta_+$ stisfies (P2)$^\vee$. This means, by definition
of $\check{\BB}_{\zeta_+}$, immediately \be
\label{eq:Metric4} \check{K}_F(\check{\BB}_{\zeta_+},\check{\BB}_{\zeta_+})\quad
\subset \quad z^{n+1} \check{\OO}_{S,0} \vspace{-0.3cm} \ee
\end{proof}

Step 10.  {\it The initial term $\zeta_+^0$ of $\zeta_+$ belong to
$z\LL_{\zeta_+}^0$.}
\begin{proof} The condition (P1) on $\zeta_+$ implies that $\zeta_+\in
\check{\BB}_{\zeta_+}$ \eqref{eq:Bzeta}. In particular, this means $\zeta_+\bmod
m^{k+1} \in \check{\BB}_{\zeta_+} /m^{k+1}\check{\BB}_{\zeta_+} \subset
z\check{\LL}_{\zeta_+}^k$. Then, taking modulo the maximal ideal $m$, we obtain
the result.
\end{proof}

Step 11.  {\it The initial term $\zeta_+^0$ of $\zeta_+$ is primitive and
homogeneous.}
\begin{proof} Primitivity: Recall that the section $\check{\BB}_{\zeta_+}$ is
given by the covariant differentiation of $\zeta_+$ \eqref{eq:Bzeta}. By
definition of the section, its projection to $\check{\Omega}_F$ is surjective,
i.e.\ $\check{\Omega}_F\subset \{\tilde{v}F\zeta_+\mid v\in \check{\TT}_S\}
\subset \check{\OO}_{C_F}\zeta_+\subset \check{\Omega}_F$.  Specializing this to
$t=0$ (i.e.\ $\bmod m$), we obtain the primitivity of $\zeta_0$.

Homogeneity : The condition (P4) on $\zeta_+$ implies
$$
\nabla_{z\partial_0}\zeta^0_+-r\zeta^0_+= \nabla_{z\partial_z}(\zeta_+|X)-r(\zeta_+|_X)=( -\nabla_E\zeta_+)|_X\in z^{-1}(\check{\BB}_{\zeta_+}\bmod m)\subset \LL^0_{\zeta_+}.
$$
\end{proof}

Step 12. {\it The initial data $ \LL_{\zeta_+}^0$ (i.e.\ \eqref{eq:Oppk} for ${k=0}$)
is a $z\partial_z$-stable opposite filtration of $\HH_f$ such that the formal
opposite filtration $\check{\LL}_{\zeta_+}^\infty$ is recovered from the initial
data $ \LL_{\zeta_+}^0$ by the evolution map 
\eqref{eq:evolve}: }
\be
\label{eq:FlatExtension3} e^{(f-F)/z} ({\LL}_{\zeta_+}^0\hat{\otimes}_{\CC}
\check{\OO}_{S,0}) \quad = \quad \check{\LL}_{\zeta_+}^\infty.  
\ee
\begin{proof} We know already that $ \LL_{\zeta_+}^0$ is a $z\partial_z$-stable
opposite filtration in $\HH_F / m\HH_F = \HH_f$ (see {\it Proof.} of Step
9.). Let us show \eqref{eq:FlatExtension3}. We first note the following.

\medskip a) For any given $a\in \LL_{\zeta_+}^0$, the equation
$\nabla_{\TT_S}\alpha=0$ for $\al\in \check{\HH}_F$ with the initial condition
$\alpha\! \mod\! m=a$ has a unique solution in $\check{\LL}_{\zeta_+}^\infty$.
\begin{proof} This essentially reduces to the existence and the uniqueness of
the (formal) solution to the system of the first order linear partial
differential equation $\nabla_{\TT_S}\alpha=0$ (for the coefficients of $\al$ in
the Laurent expantion in $z$, c.f.\ \eqref{eq:Oppk}) over $S$ for the given
initial value $a$ at $\underline{t}=0$, where one should note that the equation
is non-singular (i.e.\ the coefficient matrices are regular) on $S$. For
details, see \cite{LLS} Claim B in {\it Proof} of {\bf Lemma/Definition 4.18}.
\end{proof}

Let us return to a proof of \eqref{eq:FlatExtension3}. Let $a\in
\LL_{\zeta_+}^0$ as above. The flatness (see \eqref{eq:flatness}) of element
$e^{(f-F)/z}a \in e^{(f-F)/z}\HH_f\subset \check{\HH}_F$ implies that it is a
solution of the equation in a) satisfying the initial condition $e^{(f-F)/z}a
\bmod m=a$. Due to the uniqueness of the solution, we observe that
$e^{(f-F)/z}\LL_{\zeta_+}^0$ is contained in
$\check{\LL}_{\zeta_+}^\infty$. More precisely, for any $k\in\ZZ_{\ge0}$, we
have inclusions $e^{(f-F)/z}\LL_{\zeta_+}^0 / m^{k+1}\subset
(\check{\LL}_{\zeta_+}^\infty/ m^{k+1}\check{\LL}_{\zeta_+}^\infty)$ and hence
$e^{(f-F)/z}(\LL_{\zeta_+}^0\otimes_\CC\OO_{S,0}/m^{k+1}) \subset
(\check{\LL}_{\zeta_+}^\infty/ m^{k+1}\check{\LL}_{\zeta_+}^\infty)$. Then,
taking the limit w.r.t.\ $k\in\ZZ_{>0}$, we obtain an inclusion
$e^{(f-F)/z}({\LL}_{\zeta_+}^0\hat{\otimes}_{\CC} \check{\OO}_{S,0}) \subset
\check{\LL}_{\zeta_+}^\infty$.

On the other hand, we have splittings \eqref{eq:Splitting} and
\eqref{eq:Splitting2} of the same module $\check{\HH}^{(0)}_F$ of the same first
factor $\check{\HH}^{(0)}_F$, so that the inclusion is actually the equality
\eqref{eq:FlatExtension3}.
\end{proof}

\medskip The following Step 13.  is the final step of the proof of Theorem
\ref{Theorem}.

\medskip Step 13. {\it The compositions $\Psi\circ\Phi$ and $\Phi\circ\Psi$ are
identities on each side of \eqref{eq:PFwithoutmetric} and
\eqref{eq:PFwithmetric}, respectively.}
\begin{proof} Let $(\LL,\zeta_0)$ be a good pair. Set
$\zeta_+:=\Phi(\LL,\zeta_0)$ so that $\Psi(\zeta_+)=(\LL_{\zeta_+}^0,\zeta_+
\bmod m)$.  Recall \eqref{eq:Bzeta} and \eqref{eq:ZetaPlus}, and we get
$\check{\BB}_{\zeta_+}=\check{\HH}_F^{(0)} \cap z
e^{(f-F)/z}({\LL}\hat{\otimes}_{\CC} \check{\OO}_{S,0})$.  Restricting this
equality on the subspace $X\subset Z$ (i.e. taking $\bmod\ m$, or setting
$k=0$), we obtain $\BB_{\zeta_+}^0= \HH_f^{(0)}\cap \LL$, where the LHS
$\BB_{\zeta_+}^0:=\check{\BB}_{\zeta_+}/m\check{\BB}_{\zeta_+}$ is the section
of $\HH_f^{(0)}$ corresponding to the opposite filtration $\LL^0_{\zeta_+}$ in
the sense of Proposition \ref{SectOppo} and the RHS is the section corresponding
to $\LL$. Again due to Proposition \ref{SectOppo}, we get $\LL^0_{\zeta_+} =
\LL$.  On the other hand, $\zeta_+ \bmod m= \zeta_+ +\zeta_-\bmod m=
e^{(f-F)/z}\zeta_0 \bmod m =\zeta_0$ (here, we use Proposition \ref{zetaminus}
for the vanishing of $\zeta_- \bmod m$). That is, $\Psi(\zeta_+)=(\LL,\zeta_0)$,
and hence $\Psi\circ\Phi$ is the identity on the set of good pairs.
 
 Conversely, let us start with a formal primitive form $\zeta_+$. Then
$\Phi\circ\Psi(\zeta_+)=\Phi(\LL_{\zeta_+}^0,\zeta_+^0)$ where
$\zeta_+^0:=\zeta_+\mod m$, by definition, is the first splitting factor of
$e^{(f-F)/z}\zeta_+^0$ in $\check{\HH}_F=\check{\HH}_F^{(0)}\oplus e^{(f-F)/z}
({\LL}_{\zeta_+}^0\hat{\otimes}_{\CC} \check{\OO}_{S,0})$, and $e^{(f-F)/z}
({\LL}_{\zeta_+}^0\hat{\otimes}_{\CC} \check{\OO}_{S,0}) =
\check{\LL}_{\zeta_+}^\infty$ \eqref{eq:FlatExtension3}.  Therefore, in order to
show that it is actually equal to $\zeta_+$, i.e.\ $\Phi\circ\Psi$ is an
identity, we have to show that $\xi:=e^{(f-F)/z}\zeta_+^0-\zeta_+ \in
\check{\LL}_{\zeta_+}^\infty$. We know that $\xi$ is an element of $
z\check{\LL}_{\zeta_+}^\infty\subset \check{\HH}_F$ since $\zeta_+^0\in
z\LL_{\zeta_+}^0$ and $\zeta_+\in \check{\BB}_{\zeta_+}$ (see Step 10. and its
proof). On the other hand, we see 1) $\xi$ satisfies differential equation
$\nabla_{\TT_S} \xi \ (= - \nabla_{\TT_S} \zeta_+) \subset
z^{-1}\check{\BB}_{\zeta_+}\subset \check{\LL}_{\zeta_+}^\infty$ (recall
\eqref{eq:flatness} and \eqref{eq:Bzeta}), and 2) the initial value: $\xi \bmod
m$ is equal to zero (by the definition of $\zeta_+^0$). The solution for such
system of equations belongs to $\check{\LL}_{\zeta_+}^\infty$ ({\it Proof.} this
can be shown similar to a) in the proof of Step 12. Namely, let $\xi_0 \in
\check{\BB}_{\zeta_+}$ be the degree zero term of the Laurent expansion of $\xi$
in $z$, where we note that there is positive degree term in $z$. Then it
satisfies 1) a regular linear equation $\nabla_{\TT_S} \xi_0=0 \bmod
\check{\LL}_{\zeta_+}^\infty $, and 2) the initial value $\xi_0 \bmod m$ is
equal to $0$, and then the unique solution is $\xi_0=0$.).
\end{proof} 
\vspace{-0.4cm}
This completes a proof of Theorem (6.1).  
\end{proof}

\section{Flat structure without metric structure}

It is well-known that a primitive form induces a flat structure (Frobenius
manifold structure~\cite{Hertling,Dubrovin, Manin}) on $S$ (\cite{FromPrimitiveToFrobenius}). In the present section we confirm
that a primitive form without a metric structure induces a flat structure
without a metric structure in the sense of Sabbah (cf.~\cite{KMS1, KMS2, KMS3}), in particular, flat coordinate system on $S$ are introduced. This
is achieved by writing down explicitly the holonomicity property (P3)$^{*}$ of a primitive form $\zeta$ as in the following proposition.

\begin{Proposition} 
\label{FlatStructure} {\bf (Flat structure without metric structure.)} 

{\bf 1.}  Let $\zeta$ be a primitive form without higher residue structure (Definition \ref{PFwithoutmetric}). Then it satisfies the following {\bf (P3)}$^{**}$, which is an explicit form of {\rm (P3) and (P3)$^* $(Holonomicity)}.

\smallskip
\noindent {\rm (P3)$^{**}$ (Holonomicity)$^{**}$}: The second covariant
derivatives of $\zeta$ have the followings expressions.

\quad 1.  There exists an affine torsion free connection
$\nabla\!\!/: \TT_S\times \TT_S \to \TT_S$ on $S$ such that 
\be
\label{eq:FlatConnection} \nabla_{v_1}\nabla_{v_2}\zeta =z^{-1}\nabla_{v_1*v_2}
\zeta + \nabla_{\nabla\!/_{v_1}v_2}\zeta \quad \text{for }\ v_1,v_2 \in\TT_S .
\ee 
\qquad \qquad where ``$*$" means the product structure on $\TT_S$ given in
\eqref{eq:*}.

\smallskip \quad 2.  There is an $\OO_S$-endomorphism $N\in
End_{\OO_S}(\Omega_F)$ such that 
\be
\label{eq:Exponents} \nabla_{z\partial_z}\nabla_{v}\zeta= -z^{-1}\nabla_{E*v} \zeta + 
\nabla_{N(v)}\zeta \quad \text{for }\ v \in\TT_S , 
\ee \qquad\qquad 
where $E$ is
the Euler vector field (see Definition 2.2).

\smallskip \quad 3.  The connection $\nabla\!/$ and the endomorphism $N$ satisfy
the following relations.

\qquad i) Flatness and charge of Primitive vector field: \quad  $\nabla\!/ \partial_1 =0$ \quad and \quad $N(\partial_1)= r\partial_1$.

\qquad ii)  Integrability of $\nabla\!/$:  \quad   $[\nabla\!/ ,\nabla\!/] =0$,

\qquad iii) Symmetry of $\nabla\!/ *$:  \ Set $T(u,v,w):=\nabla\!/_u(v*w)-(\nabla\!/_uv)*w-v*(\nabla\!/_u w)$ for $u,v,w\in\TT_S$.

\qquad \quad   Then, $T$ is a symmetric $\OO_S$-tensor of type $(3,1)$, i.e.\ $T\in Hom_{\OO_S}(S^3(\TT_S), \TT_S)$.

\qquad iv) Horizontality of $N$:  \quad  $\nabla\!/N=0$.

\qquad v) Exponents:  \quad   $N(v) =(r+1)\cdot v -\nabla\!/ _v E \quad \text{for} \ v\in \TT_S$,

\qquad vi) Flatness of $E$:  \quad $\nabla\!/_w\nabla\!/_v E \ = \ \nabla\!/_{\nabla\!/_w v} E$ \quad for $v,w\in \TT_S$.

\qquad vii) Homogeneity of $*$-product:  \quad $[E,v*w]\ = \ v*[E,w] \ + \ [E,v]*w+ v*w$ \quad for $v,w\in \TT_S$.

\medskip {\bf 2.}  Let $\zeta$ be a formal primitive form without metric structure (Definition \ref{FPFwithoutmetric}). Then it satisfies the following  {\rm (P3)}$^{**\vee}$ (a explicit form of {\rm (P3)$^\vee$ ) and (P3)$^{* \vee}$ (Holonomicity)}), where

\medskip
\noindent {\rm (P3)$^{**\vee}$ (Holonomicity)$^{**\vee}$}: Replace $\TT_S$ in
{\rm (P3)$^{**}$} by $\check{\TT}_S$.
\end{Proposition}

\begin{proof} We prove only {\bf 1.}, since {\bf 2.} can be proven parallery to
the case for {\bf 1.}


\noindent 
{\it Proofs of }1. {\it and} 2.:  According to (P3)$^*$, we set two expansions:
$$
\begin{array}{rlll} \nabla_{v_1}\nabla_{v_2}\zeta & =& z^{-1}\nabla_{A(v_1,v_2)}
\zeta + \nabla_{B(v_1,v_2)}\zeta \quad & \text{for }\ v_1,v_2 \in\TT_S \\ \\
\nabla_{z\partial_z}\nabla_{v}\zeta & = & z^{-1}\nabla_{M(v)} \zeta +
\nabla_N(v)\zeta \quad & \text{for }\ v \in\TT_S ,
\end{array}
$$
where $A, B:\TT_S\times \TT_S \to \TT_S$ are $\CC$-bilinear morphism in the variables $v_1,v_2$ and $\OO_S$-linear morphisms in the variable $v_1$, and $M,N:\TT_S\to \TT_S$ are
$\OO_S$-linear morphisms in the variable $v$.

In view of \eqref{eq:initial} and Kodaira-Spencer morphism \eqref{eq:KS2} with the $*$-product \eqref{eq:*}, the first terms of the above expansions are calculated as  $A(v_1,v_2)=v_1* v_2$ and $M(v)= - E* v$, respectively.  
Further more inserting the above expansions to the integrability relation and
the Leibniz rule of the Gauss-Manin connection $\nabla$, 
we see that $B(v_1,v_2)$ is given as the covariant differentiation $\nabla\!/_{v_1}v_2$ of the tangent vector $v_2$ by $v_1$  for an affine torsion free connection $\nabla\!/$ on $S$:
\begin{equation*}
  \nabla\!/v w-\nabla\!/w v=[v,w] \quad \text{\quad  for \ }v,w\in \TT_S.
\end{equation*}
Thus we obtain the expressions \eqref{eq:FlatConnection} and \eqref{eq:Exponents}.

\medskip
\noindent 
{\it Proof of } 3. : 
The covariant derivation by $v\in \TT_S$ of the $\partial_1$-triviality (P1) $z\nabla_{\partial_1}\zeta=\zeta$  implies the horizontality $\nabla\!/_v\partial_1=0$ of the primitive vector field. 
The covariant derivation by $z\partial_z$ of the  $\partial_1$-triviality (P1) $z\nabla_{\partial_1}\zeta=\zeta$ implies the equality $z\nabla_{\partial_z}\zeta=-\nabla_E\zeta+z\nabla_{N(\partial_1)}\zeta$. Comparing this with the homogeneity axiom (P3) of $\zeta$, we get $\nabla_{r\partial_1}\zeta=\nabla_{N(\partial_1)}\zeta$ and, therefore, $N(\partial_1)=r\partial_1$. That is, $\partial_1$ is an eigenvector of $N$ belonging to the eigenvalue $r$, and ({\bf 1.} 3. i)) is shown.
\be 
\label{eq:partial1}
\nabla\!/ \partial_1=0  \quad \text{and}\quad N(\partial_1)=r\partial_1
\ee

The integrability relation $[\nabla_{v_1},\nabla_{v_2}]=\nabla_{[v_1,v_2]}$ acting  
 on
$\nabla_{w}\zeta$ for $w\in \TT_S$ implies
 the integrability
$[\nabla\!/_{v_1},\nabla\!/_{v_2}]=\nabla\!/_{[v_1,v_2]}$ of $\nabla\!/$ ({\bf 1.} 3. ii))
together with the relation:
\begin{equation}
\label{eq:*nabla}
[v_1,v_2]*v_3= \nabla\!/_{v_1}(v_2*v_3)+v_1*\nabla\!/_{v_2}v_3 -
\nabla\!/_{v_2}(v_1*v_3)-v_2*\nabla\!/_{v_1}v_3.
\end{equation}
Let us show that the expression:
\be
\label{eqq:Tensor}
T(u,v,w):=\nabla\!_u(v*w)-(\nabla\!/_uv)*w-v*(\nabla\!/_u w) \quad \text{for} \quad u,v,w\in\TT_S
\ee
 is a symmetric $\OO_S$-tensor of type $(3,1)$ (i.e.\ a $\mathfrak{S}_3$-symmetric $\OO_S$-tri-linear map $\TT_S^{\otimes3}\to \TT_S$). The $\OO_S$-linearity on the variable $u$ is obvious. The symmetry by the permutation of $v$ and $w$  is also obvious. Therefore, we have only to show that the symmetry by the permutation of $u$ and $v$. However, it is immediate to see 
 $T(v_1,v_2,v_3)-T(v_2,v_1,v_3)=0$ due to \eqref{eq:*nabla}.  
 Thus,  ({\bf 1.} 3. iii)) is shown.

\medskip
The integrability (commutativity) relation $[\nabla_{z\partial_z},\nabla_v]=0$ acting on $\nabla_w\zeta$ for $w\in \TT_S$ implies
 the horizontality $\nabla\!/ N=0$ of $N$ ({\bf 1.} 3. iv)) together with the relation: 
\be
\label{eq:*N}
N(v*w)-v*N(w)= E*\nabla\!/_vw -  \nabla\!/_v(E*w) + v*w
\ee

Inserting $w=\partial_1$ in \eqref{eq:*N} and applying \eqref{eq:partial1}, we obtain the expression ({\bf 1.} 3. v))  of $N$:
\be
\label{eq:N}
N(v) \quad = \quad (r+1)\cdot v \ - \ \nabla\!/_v E .
\ee

Applying the expression \eqref{eq:N} back to \eqref{eq:*N}, we obtain the relation
\be
\label{eq:deg*1}
\nabla\!/_{v*w} E-v*\nabla\!/_w E+ E*\nabla\!/_v w  +  v*w \quad = \quad \nabla\!/_v(E*w).
\ee

In \eqref{eq:*nabla}, insert $v_1=E$, $v_2=v$ and $v_3=w$ and take the difference  with \eqref{eq:deg*1} so that we obtain
$$
\big(\nabla\!/_E (v*w)-\nabla\!/_{v*w}E\big) -v*\big(\nabla\!/_Ew-\nabla\!/_w E\big) - [E,v]*w \quad = \quad v*w.
$$
Remembering that $\nabla\!/$ is torsion free, this means that the homogeneity  ({\bf 1.} 3. vii)) of the $*$-product.
\be
\label{eq:deg*}
[E,v*w]\quad = \quad v*[E,w] \ + \ [E,v]*w + v*w.
\ee
This completes a proof of Proposition \ref{FlatStructure}.
\end{proof}

\begin{remark} \label{biFlatStructure}
  Let the Euler vector field $E$ be generically invertible as an element of the Frobenius algebra $\TT_S$ and
  we denote the inverse element $E^{-1} \in \TT_S$. Then
  the structure in the Proposition \ref{FlatStructure} is actually a bi-flat structure in a sense that there is another
  affine torsion free compatible connection $\nablab^* \; : \; \TT_S \times \TT_S \to \TT_S$ (defined by the Euler vector field) such that
  $\nablab^*_v w = \nablab_v E^{-1} * w$. 
\end{remark}

\begin{remark} Including the case when $\zeta_+$ is a formal primitive form (with or without higher residue structure),
recall \eqref{eq:FormalNabla} that the Gauss-Manin connection $\check{\nabla}$ acts on the trivial bundle $\check{\HH}_F$. However, the formulae \eqref{eq:FlatConnection} and \eqref{eq:Exponents} imply that the $\OO_S[z^{\pm1}]$-free submodule $\BB_{\zeta_+}[z^{\pm1}]$, where we set $\BB_{\zeta_+}:= Im( z\nabla \zeta_+:
\TT_{S,0}\to \check{\HH}^{(0)}_F) = z\nabla_{\TT_{S,0}}\zeta_+$, is invariant under the actions of $\nabla_{\partial_z}$ and $\nabla_v$ for $v\in\TT_{S,0}$.  In this sense, the connection is extended to a meromorphic connection on $\mathbf{P}^1\times S'$ where $S'$ is a neighborhood of the origin $0\in S$ and $\mathbf{P}^1$ is the projective line of inhomogeneous coordinate $z$. For brevity, we shall call $S'$ again $S$. 
So we obtain
\be
\label{eq:Flat Structure}
\nabla\ :\  \TT_{\mathbf{P}^1\times S}  \times \BB_{\zeta_+}[z^{\pm1}] \quad \longrightarrow \quad \BB_{\zeta_+}[z^{\pm1}] 
\ee
where we may regard $\BB_{\zeta_+}[z^{\pm1}]$ as a finite rank bundle $\pi^*\TT_S$ on $\mathbf{P}^1\times S$  for $\pi: \mathbf{P}^1\times S \to S$.
\end{remark}

\noindent
{\bf Flat coordinates of $S$.}  The integrability of the connection $\nabla\!/$ implies that the unfolding parameter space $S$ carries an affine linear structure (\cite{S1}). Namely, {\it torsion-freeness of $\nabla\!/$ implies that the horizontal sections of the dual connection $\nabla\!/^*$ of $\nabla\!/$ on the cotangent bundle $\Omega_S^1$ of $S$ are closed forms. Therefore, the following second order differential equation on a function $u$ on $S$:
\footnote{
Using tangent-cotangent duality, the equation is equivalent to a system of equations
$ (vw-\nabla\!/_vw) u=0$ for $v,w\in \TT_S$.
}
\be
\label{eq:F-Coordinate1}
\nabla\!/ ^* du\ =\quad 0
\ee
is integrable and has (locally) $\mu+1$ linearly independent solution (including constant functions)}. We shall call them the (local affine linear) {\it flat coordinate system} of $S$ w.r.t. the primitive form $\zeta$. Using the natural pairing between tangent and cotangent bundles, the equation \eqref {eq:F-Coordinate1} is equivalent to the vanishing of $\langle v_2,\nabla\!/^*_{v_1} du\rangle= v_1\langle v_2, du\rangle- \langle \nabla\!/_{v_1} v_2, du\rangle$ for all $v_1,v_2\in \TT_S$. That is,
\be
\label{eq:F-Coordinate2}
 v_1 v_2 u-  (\nabla\!/_{v_1} v_2) u \ = \ 0   \qquad \forall v_1,v_2\in \TT_S .
 \ee
Actually, the system \eqref{eq:F-Coordinate2} is integrable due to the torsion freeness of $\nabla\!/$:  
$\big( v_1 v_2 -  (\nabla\!/_{v_1} v_2)\big) -\big(  v_2 v_1 -  (\nabla\!/_{v_2} v_1) \big)=[v_1,v_2]- (\nabla\!/_{v_1} v_2-\nabla\!/_{v_2} v_1) =0$.

\medskip
For a sake of completeness, we add a few classical well-known results on the case when $\zeta$ is a primitive form with metric structure. Recall that Kodaira-Spencer morphism \eqref{eq:KS} has identified the tangent bundle of $S$ with $\pi_*\OO_C$. A choice of a primitive form (without metric structure) identify them with the vector bundle $\Omega_F$ over $S$ due to (P0), where we recall that the bundle $\Omega_F$ admits an inner product $\eta$ \eqref{eq:HRdeg}. We denote the pull back metric $\eta_{\zeta}:=(z\nabla\zeta)^{-1}\zeta$ again simply by $\eta$ and regard it as a metric on  the tangent bundle $\TT_S$ of $S$.

\begin{Proposition} 
\label{MetricCondition}
Let $\zeta$ be a primitive form satisfying the metric condition {\rm (P2)} and let $\eta$ be the metric on $\TT_S$ induced from the residue pairing \eqref{eq:HRdeg}. Then, the flat connection $\nabla\!/$ and the endomorphism $N$ in Proposition \ref{FlatStructure} satisfy further the  following relations:

\qquad i) Horizontality of $\eta$: \quad  $\nabla\!/ \eta =0$.

\qquad ii) Self-adjointness of $*$-product: \quad  $\eta(u*v,w)\ =\ \eta(u,v*w)$ \ for $u,v,w,\in \TT_S$.

\qquad  iii) Duality of exponents: \quad $N\ +\  N^*\ =\ 1$ where $N^*$ is the adjoint of $N$ w.r.t.\ $\eta$.
\end{Proposition}
\begin{proof} According to the above convention of the metric $\eta$, the condition (P2) implies the equality: 
$K_F(\nabla_v\zeta,\nabla_w\zeta)= \eta(v,w) z^{n-2}$ for $v,w\in \TT_S$. Let us apply this to the formulae \eqref{eq:HRv} and \eqref{eq:HRz}, where we choose $\zeta_i=\nabla_{v_i}\zeta$ for $v_i\in \TT_S$ ($i=1,2$) and  use the expansion formulae \eqref{eq:FlatConnection} and \eqref{eq:Exponents}.  Then, for the first case  \eqref{eq:HRv}, we obtain:
\begin{equation}
v \eta(v_1,v_2) z^{n-2}=\eta(v*v_1,v_2)z^{n-3}+\eta(\nabla\!/_vv_1,v_2)z^{n-2}-\eta(v_1,v*v_2)z^{n-3}+\eta(v_1,\nabla\!/_v v_2)z^{n-2}. \label{metric1}
\end{equation}
The relation of coefficients of $z^{n-2}$ implies i) the horizontality $\nabla\!/\eta=0$, and the relation of coefficients of $z^{n-3}$ implies ii) the self-adjointness of the $*$-product.  
For the second case  \eqref{eq:HRz}, we obtain:
\begin{equation} \label{metric2}
(n-2 +*)  \eta(v_1,v_2) z^{n-2}=\eta(E*v_1,v_2)z^{n-3}+\eta(N(v_1),v_2)z^{n-2}-\eta(v_1,E*v_2)z^{n-3}+\eta(v_1,N(v_2))z^{n-2}.
\end{equation}
The relation of coefficients of $z^{n-2}$ implies iii) the duality $N+N^*=1$, and the relation of coefficients of $z^{n-3}$ is a special case of ii) and does not give any new information.  
\end{proof}

Actually, the residue pairing defines a metric on $\TT_S$ via Kodaira-Spencer map even when the metric condition is not satisfied. But this
metric is not flat in general.

\begin{Proposition}
  Let $\zeta$ be a primitive form without metric structure. Let $\eta$ be a metric induced on $\TT_S$ via the map $v \to z\nabla_v \zeta
  \in \BB \simeq \Omega_F$. Let $(\TT_S)^0 \subset \TT_S$ be a subbundle of vector fields for which the condition $(P2)$ is satisfied,
  that is if $v,w \in (\TT_S)^0$ then $K_F(\nabla_v \zeta, \nabla_w \zeta) \in \CC z^{n-2}$.
 Then the $*$-product is symmetric with respect to $\eta$:\quad  $\eta(u*v,w)\ =\ \eta(u,v*w)$ \ for $u,v,w,\in \TT_S$.
 
\end{Proposition}
\begin{proof}
  In general we have $K_F(\nabla_v \zeta, \nabla_w \zeta) = \eta(v,w)z^{n-2}+O(z^{n-1})$. We consider the relation~\eqref{metric1} modulo the
  terms of order $O(z^{n-2})$:
  \begin{equation}
    0=\eta(v*v_1,v_2)z^{n-3}-\eta(v_1,v*v_2)z^{n-3},
  \end{equation}
  which proves the statement.
\end{proof}

\noindent
{\bf Flat potential on S.}
As a consequence the metric property  of $\nabla\!/$, we are now able to define a potential function \cite{S1}. Namely, {\it the following third order differential equation on a function $F$ on $S$ 
\be
\label{eq:FlatPotential}
uvw F \ =\quad  \eta(u*v,w) \quad  \text{ for } \ u,v,w \in  Kernel(\nabla\!/)
\ee
is solvable.} Here we denote by $Kernel(\nabla\!/)$ the space of horizontal sections of $\nabla\!/$ (note that elements in $Kernel(\nabla\!/)$ are commutative as derivations due to the torsion-free-ness of $\nabla\!/$).  ({\it Proof of the integrability of \eqref{eq:FlatPotential}:} The 3-form in RHS of the equation \eqref{eq:FlatPotential} is symmetric in three vector fields  $u,v,w$ on $S$ (Proposition \ref{MetricCondition} ii)). Let show that the 4-form $x \, \eta(u*v,w)$ for 4 flat vector fields $u,v,w,x \in Kernel(\nabla\!/)$ is symmetric. Let us specialize the symmetric tensor $T$ (Proposition \ref{FlatStructure} {\bf 1.} 1. iii)) to flat vector fields $u,v,w\in Kernel(\nabla\!/)$.  Then, we see that $\nabla\!/_u(v*w)$ is symmetric in the three varibles. On the other hand, using the horizontality of $\eta$ (Proposition 
\ref{MetricCondition} i)), we see $x \, \eta(u*v,w)=\eta(\nabla\!/_x(u*v),w)+\eta(u*v,\nabla\!/_xw)=\eta(\nabla\!/_x(u*v),w)$.  LHS of this equality was already symmetric in the variables $u,v,w$, and RHS of the equality is now symmetric in three variables $x,u,v$. These together imply that the 4-tensor on $u,v,w,x \in Kernel(\nabla\!/)$ is symmetric. $\Box$)

\begin{remark} 
  Morihiko Saito \cite{MSa1} introduced the concept of very good section.
   Recently, Todor Milnov gave a characterization of a very good section in terms of the flat structure as follows.

\begin{Proposition} (\cite{M})
  A section $\LL$
   is very good if and only if 
\be
\label{eq:Verygood}
[N, E]\quad  =\quad -E.
\ee
\end{Proposition}

One observes that the formula \eqref{eq:Verygood} does not involve the metric structure (either the higher or the first residue pairings). Therefore, one can formally define a very good section without metric structure, if a good section without the metric condition satisfies \eqref{eq:Verygood}. Using the flat structure in Proposition \ref{FlatStructure}, we can rewrite the condition \eqref{eq:Verygood} only in terms of $E$ as follows.
\be
\label{eq:Verygood2}
\nabla\!/_E(E*a)-E*\nabla\!/_Ea \ =\ -E*a   \quad \text{for } \  a\in \TT_S .
\ee
\begin{proof}
Applying Proposition {\bf 1.} 3. v), LHS of the  condition \eqref{eq:Verygood} implies
$$
N(E*a)-E*N(a) = \nabla\!/_{E*a}E -E*\nabla\!/_aE= [E*a,E]+\nabla\!/_E(E*a) - E*([a,E]+\nabla\!/_E a) =\nabla\!/_E(E*a)-E*\nabla\!/_E a 
$$
\end{proof}
\end{remark}




\paragraph{Saito structures without metric and flat F-manifolds}

Now we can compare our construction with the definition of the Saito structure
without metric by Sabbah~\cite{Sabbah}.

Let $M$ be a complex analytic manifold, $TM$ its tangent bundle, $\Theta_M$ the
sheaf of holomorphic vector fields and $\Omega^1_M$ the sheaf of holomorphic
1-forms.

\begin{definition}(Sabbah) A \textit{Saito structure without metric} on $M$ consists of
the following data:
    \begin{enumerate}
        \item A torsionless flat connection $\nabla$ on $TM$.
        \item A 1-form $\Phi$ with values in $\mathrm{End}(TM)$, namely a
section of the sheaf $\mathrm{End}(\Theta_M)\otimes_{\OO_M}\Omega^1_M$, which is
symmetric when considered as a bilinear map $\Theta_M\otimes_{\OO_S}\Theta_M \to
\Theta_M$;
        \item Two global sections of $\Theta_M$ (vector fields) $e$ and
$\mathfrak{E}$ respectively called unit field and Euler field such that the following two conditions are satisfied:
    \end{enumerate}
    \begin{enumerate}
        \item The meromorphic connection $\boldsymbol{\nabla}$ defined on the vector bundle
$\pi^* TM$ on $\mathbb{A}^1 \times M$ by 
\begin{equation}
  \boldsymbol{\nabla} = \pi^* \nabla +
  \frac{\pi^* \Phi}\tau - \left(\frac{\Phi(\mathfrak{E})}\tau + 
  \nabla \mathfrak{E}  \right) \frac{\dd \tau}\tau 
\end{equation}
is integrable.
        \item The vector field $e$ is $\nablab $-horizontal and $\Phi(e)= -Id$.
    \end{enumerate}
\end{definition} We consider a tuple $(M, \nablab, \Phi, e, \mathfrak{E})$
satisfying the conditions above as a Saito structure without metric.

Given a universal unfolding $F$ over a frame $(Z,S,p)$ the multiplication
structure $* \; : \; \TT_S \times \TT_S \to \TT_S$ defines a symmetric tensor
$C$ in $\mathrm{End}(\TT_S) \otimes_{\OO_S} \Omega_S^1$ via $C(v,w) = v*w, \; v,w \in \TT_S$,
since the multiplication is $\OO_S$-linear.

\begin{Proposition} Let $\zeta$ be a primitive form without metric in a sense of
Definition~\ref{PFwithoutmetric} .Then the tuple $(S, \nablab, -C, \pd_1, E)$ defines a Saito structure without
metric.
\end{Proposition}
\begin{proof} 

  1. The connection $\nablab \; : \; \TT_S \times \TT_S \to \TT_S$ is flat and torsionless
  by the proposition~\ref{FlatStructure}. \\

  2. The 1-form $-C$ is symmetric since the multiplication $*$ is commutative. \\

  Condition 1: The connection $\boldsymbol{\nabla}$ reads:
  \begin{equation} \label{eq:SabbahConnection}
    \boldsymbol{\nabla} = \pi^* \nablab -\frac{\pi^* C}{\tau}+\left(\frac{C(E)}{\tau} + \nablab E \right) \frac{\dd \tau}{\tau}.
  \end{equation}
  The formula~\eqref{eq:SabbahConnection} means that for $v,w \in \TT_S \otimes_{\OO_S} \OO_{S \times \mathbb{A}^1}$:
  \begin{equation} \label{eq:SabbahConnection1}
    \begin{aligned}
      &\boldsymbol{\nabla}_v w =  \nablab_v w - \tau^{-1} \, v*w, \\
      &\boldsymbol{\nabla}_{\tau \pd_{\tau}} w = \tau \pd_{\tau} w - \nablab_w E + \tau^{-1} \, E * w.
    \end{aligned}
  \end{equation}
  The integrability conditions read for all $v,u,w \in \TT_S \subset \TT_S \otimes_{\OO_S} \OO_{S \times \mathbb{A}^1},
  \; [v,w]=0$:
  \begin{multline}
    (\nablab_v \nablab_w-\nablab_w\nablab_v) u = 0, \;\;\;  \nablab_v (w*u) + v*\nablab_w u =  \\
    = \nablab_w (v*u) + w*\nablab_v u, \;\;\; v*(w*u)-w*(v*u)=0
  \end{multline}
  from the compatibility of the first line of~\eqref{eq:SabbahConnection1} with itself.
  The first equation holds since $\nablab$ is flat (1. from the proposition~\ref{FlatStructure}),
  the second equation is the symmetry condition 3. iii) from~\ref{FlatStructure} and the third one
  is associativity and commutativity of the $*$-product.
  
  The following condition (together with the already checked) imply the compatibility of the first line
  of~\eqref{eq:SabbahConnection1} with the second one:
  \begin{equation}
    w*\nabla_u E + \nablab_w (E*u) = E*\nablab_w u + \nablab_w (E*u)+w*u+\nablab_{w*u} E.
  \end{equation}
  This equation is nothing but the equation~\eqref{eq:deg*1}.

  To show property 2 we note that $-C(\pd_1)=-Id$ by the primitivity property
  of $\pd_1$. Flatness of $\pd_1$ is 3. i) of the proposition~\ref{FlatStructure}.
\end{proof}

\begin{remark}
  The connection~\eqref{eq:SabbahConnection} from the Saito structures without metric is essentially the Gauss-Manin connection.
  More precisely, a primitive form $\zeta$ defines a $\OO_S$-module isomorphism $\TT_S \to \BB$
  via
  $v \to z v \zeta$. Let $\tau = -z$. Then the connection~\eqref{eq:SabbahConnection} defines a connection on $\BB((z)) \simeq \HH_F$ via
  the isomorphism $\TT_S \to \BB$. Using the formula~\eqref{eq:SabbahConnection1} we compute
  \begin{equation} \label{eq:SabbahConnection2}
    \begin{aligned}
      &\boldsymbol{\nabla}_v \nabla_w \zeta =  \nabla_{\nablab_v w} \zeta + z^{-1} \, \nabla_{v*w} \zeta, \\
      &\boldsymbol{\nabla}_{z \pd_{z}} \nabla_w \zeta = z \pd_{z} \nabla_w \zeta - \nabla_{\nablab_w E} \zeta - z^{-1} \, \nabla_{E * w} \zeta.
    \end{aligned}
  \end{equation}
  The first line is~\eqref{eq:FlatConnection} if we identify $\boldsymbol{\nabla} = \nabla$. The right hand side of the second line
  differs from the right hand side of~\eqref{eq:Exponents} by the summand $r \cdot \nabla_w \zeta$.
\end{remark}

A flat structure without metric also defines a structure
of the flat F-manifold~\cite{Manin} on $S$ which is a similar notion.

\begin{definition}(Manin)
  An \textit{analytic F-manifold with compatible flat structure} (or \textit{flat F-manifold}) is a tuple $(M, \circ, \nabla, e)$,
  where $M$ is a complex manifold, $\circ \; : \; \TT_M \times_{\OO_M} \TT_M \to TM$ is a product on the tangent
  spaces depending holomorphiclly on the point in $M$, $\nabla$ is a connection on $TM$ and $e$
  is a distinguished vector field satisfying the properties:
  \begin{enumerate}
  \item The one-parameter family of connections
    $\nabla - \lambda \circ$ is flat and torsionless for any $\lambda \in \CC$.
  \item The vector field $e$ is a unit of the product and is flat $\nabla e = 0$.
  \end{enumerate}

  An F-manifold $(M, \circ, \nabla, e)$ is called conformal if it is equipped with
  a vector field $E \in \TT_M$ called the Euler vector field such that $\nabla\nabla E = 0$,
  $[e, E] = e$ and $\LL_E(\circ) = \circ$.
\end{definition}

\begin{lemma}
  Let $\zeta$ be a primitive form without metric structure. Then $(S, *, \nablab, \pd_1, E)$
  where $*$ is the multiplication structure on $\TT_S$ defined in~\eqref{eq:*} is a
  conformal flat F-manifold.
\end{lemma}
\begin{proof}
  Let us first that a primitive form without metric defines a flat F-structure.

  1. Flatness and torsionless of the connection $\nablab - \lambda *$ follow from
  the flatness and torsionless of the $\TT_S$-part of the Gauss-Manin connection
  or from the first line of~\eqref{eq:SabbahConnection1}.

  2. Vector field $\pd_1$ is a unit of the product $*$ and is flat due to the properties
  of the flat structure without metric.\\

  Properties $\nablab \nablab E = 0$ and homogeneity of the $*$-product are is 3. vi) and 3. vii)
  of the proposition~\ref{FlatStructure} correspondingly.

  To show the property $[\pd_1, E] = \pd_1$ we notice that $[\pd_1,E] = \nablab_{\pd_1} E-\nablab_E \pd_1 = \nablab_{\pd_1} E$.
  Now we use
  \begin{equation}
    r \pd_1 = N(\pd_1) = (r+1) \pd_1 - \nablab_{\pd_1} E \implies \nablab_{\pd_1} E = \pd_1.
  \end{equation}
    
\end{proof}

\section*{Acknowledgements}

This work was partially supported by 
JP18H01116 (the second named author),
OIST unit budget (Mathematical and Theoretical Physics Unit).
The first named author is grateful to IPMU, OIST and Shinobu Hikami for hospitality.

\printbibliography

\end{document}